\documentclass[11pt]{article}

\usepackage[a4paper,margin=1in]{geometry}

\usepackage[T1]{fontenc}
\usepackage[utf8]{inputenc}
\usepackage{lmodern}
\IfFileExists{microtype.sty}{\usepackage{microtype}}{}

\usepackage{amsmath,amssymb,amsthm}

\usepackage{algorithm}
\usepackage{algpseudocode}
\providecommand{\Switch}[1]{\State \textbf{switch} #1}
\providecommand{\EndSwitch}{}
\providecommand{\Case}[1]{\State \textbf{case} #1}
\providecommand{\EndCase}{}

\usepackage{tikz}
\usetikzlibrary{arrows.meta,positioning,calc,decorations.pathreplacing,shapes.geometric,fit}
\IfFileExists{pgfplots.sty}{%
  \usepackage{pgfplots}
  \pgfplotsset{compat=1.18}
}{%
  \newenvironment{axis}[1][]{\par\small\itshape[pgfplots unavailable: plot omitted]\par}{}
  \newcommand{\addplot}[2][]{}
  \newcommand{\addlegendentry}[1]{}
  \newcommand{\pgfplotsset}[1]{}
}

\usepackage{booktabs}
\usepackage{multirow}
\usepackage{array}
\usepackage{tabularx}

\usepackage{listings}
\usepackage[colorlinks=true,linkcolor=blue!60!black,citecolor=blue!60!black,urlcolor=blue!60!black]{hyperref}
\IfFileExists{cleveref.sty}{%
  \usepackage[capitalise,noabbrev]{cleveref}
}{%
  \newcommand{\cref}[1]{\ref{##1}}
  \newcommand{\Cref}[1]{\ref{##1}}
}

\usepackage[numbers,sort&compress]{natbib}

\usepackage{enumitem}
\usepackage{xcolor}
\usepackage{caption}
\usepackage{subcaption}
\usepackage{xspace}

\newtheorem{theorem}{Theorem}[section]

\newtheorem{proposition}[theorem]{Proposition}

\theoremstyle{definition}
\newtheorem{definition}[theorem]{Definition}
\theoremstyle{remark}

\newcommand{\sys}{\textsc{ACE Runtime}\xspace}
\newcommand{\acegf}{\textsc{ACE-GF}\xspace}
\newcommand{\zkace}{\textsc{ZK-ACE}\xspace}
\newcommand{\idcom}{\mathsf{id\_com}}
\newcommand{\REV}{\mathsf{REV}}
\newcommand{\HKDF}{\mathsf{HKDF}}

\title{\textbf{ACE Runtime: Identity-Rooted Accounts with Installed-Key Authorization\\and Aggregate Finality Certificates}}

\author{
  \texttt{Jian Sheng Wang} \\
  \texttt{jason@acechain.io} \\
  Montreal, Canada
}
\date{November 24, 2026}

\begin{document}
\emergencystretch=1.5em   
\maketitle

\begin{abstract}
Existing high-performance blockchains such as Solana hard-code per-transaction
signature verification on the critical execution path, causing $O(N)$ signature
checking to become the primary throughput bottleneck and a major source of
hardware cost.  This paper presents \sys, a blockchain execution layer built
upon the identity--authorization separation primitive known as \acegf.
An account's identifier is a hash of its identity root; the keys that
authorize its transactions are installed in an on-chain key table and can
be Ed25519 or ML-DSA-44, and rotated, without changing the identifier.
Account birth uses the \zkace birth transaction: the user's client proves
knowledge of the identity root, and validators verify that proof before
installing the first key.  Ordinary transactions carry one signature under
an installed key, which validators verify natively before execution and
voting, so soft finality (${\sim}400$\,ms) certifies authorization.  A
three-phase \emph{Attest--Execute--Prove} pipeline adds an optional
aggregate proof, produced off the critical path from public block data
only, that lets light clients and bridges check every authorization of a
block with one verification.  The proof only delays or advances hard
finality; it never rolls back a soft-final block, and provers hold no user
secret.  Per-transaction authorization data is 113\,B with Ed25519 and
2{,}469\,B with ML-DSA-44, because public keys are installed once rather
than carried per transaction.  An analytical model on a 16-core server
with cited per-operation costs projects about 16{,}000 transactions per
second (TPS) with Ed25519 keys and 14{,}000 TPS with ML-DSA-44 keys when
every transaction conflicts and execution is serial, and CPU ceilings of
about 107{,}000 and 48{,}000 TPS (the latter set by 1\,Gbps bandwidth)
when execution parallelizes across all cores.  The aggregation proof system is not implemented or measured, so
hard-finality latency is bounded by an aggregation time that this paper
does not measure.
\end{abstract}

\medskip
\noindent\textbf{Keywords:} blockchain runtime, zero-knowledge proofs,
identity--authorization separation, installed-key authorization, Groth16,
attestation, ML-DSA-44

\section{Introduction}
\label{sec:intro}

High-performance layer-one blockchains now reach high transaction
throughput.  Solana~\cite{solana-whitepaper}, for example,
achieves 400\,ms block times and processes thousands of transactions per
second by employing a pipelined Transaction Processing Unit (TPU) that
offloads Ed25519 signature verification to GPUs.  Nevertheless, the
fundamental design of \emph{one cryptographic signature per transaction}
imposes an $O(N)$ verification cost that grows linearly with block size.
This cost causes three problems:

\begin{enumerate}[leftmargin=*]
\item \textbf{Throughput ceiling.}  Even with GPU-batched verification at
  ${\sim}2\;\mu\mathrm{s}$ per signature, a block containing 100{,}000
  transactions requires 200\,ms of SigVerify alone, consuming half the
  block interval.

\item \textbf{Hardware cost.}  Every Solana validator, not only block
  producers, must operate a GPU capable of batched signature verification,
  raising the minimum hardware budget to approximately \$7{,}500 and
  monthly operating costs to \$1{,}000--\$3{,}000.

\item \textbf{Post-quantum vulnerability.}  Transitioning to post-quantum
  signature schemes such as ML-DSA-44~\cite{nist-pqc} would inflate
  per-transaction authorization data from 96 bytes (Ed25519 signature plus
  public key) to ${\sim}3{,}732$ bytes, increasing block sizes by over
  $38\times$ and rendering the current architecture infeasible.
\end{enumerate}

\subsection{Key Insight: Identity--Authorization Separation}

The \acegf framework~\cite{ace-gf} introduces a separation between
\emph{identity} and \emph{authorization path}.  The on-chain identifier
$\idcom$ is a one-way image of a root entropy value ($\REV$); the keys
that authorize transactions are entries of an installed-key table indexed
by $\idcom$, of any supported algorithm, and can be added or rotated
without changing $\idcom$.  Zero-knowledge proofs are used only where
their cost is small relative to what they buy.  At account birth, a
\zkace proof~\cite{zk-ace} generated on the user's client shows knowledge
of the $\REV$ behind $\idcom$ and binds the first installed key; this is
paid once per account.  After birth, each transaction carries one
signature under an installed key, verified natively, as \zkace
prescribes for the transfer hot path.

Validators therefore still verify one signature per transaction, and
\sys does not remove that cost from them.  It changes three things.  The
signature scheme is a per-key property, so an account moves to
ML-DSA-44 by rotation and keeps its identifier.  Public keys are installed
once and are not carried per transaction.  And an optional aggregate
proof over the public data of a block lets an external verifier check
every authorization in the block at $O(1)$ cost, without the proof
affecting the validators' critical path or finality decisions.

\subsection{Contributions}

This paper makes the following contributions:

\begin{enumerate}[leftmargin=*]
\item \textbf{Attest--Execute--Prove pipeline.}  We design a three-phase
  execution pipeline in which transactions signed under installed keys
  are verified natively and executed on the critical path, and accounts
  are born through the \zkace birth transaction, while an aggregate
  proof over public block data is produced off the critical path,
  targeting 400\,ms block times (\cref{sec:pipeline} and \cref{sec:attestation}).

\item \textbf{Proof pipelining.}  We present a proof generation scheme
  with tree-structured recursive aggregation that any builder can run
  from public data to produce a single Groth16~\cite{groth16} proof per
  block, under an aggregation-cost assumption that we state and do not
  measure (\cref{sec:proof-pipeline}).

\item \textbf{Two-tier finality model.}  We formalize a finality model
  with \emph{soft} finality (BFT voting over natively verified
  authorization, ${\sim}400\;\mathrm{ms}$) and \emph{hard} finality
  (verification of the aggregate certificate, bounded by the unmeasured
  aggregation time), in which a missing certificate affects liveness only
  (\cref{sec:finality}).

\item \textbf{Quantitative evaluation.}  We provide
  prototype microbenchmarks and model-based analytical comparisons
  covering per-transaction authorization data (113\,B with Ed25519,
  2{,}469\,B with ML-DSA-44), bandwidth-limited throughput, and $O(1)$
  certificate verification for external verifiers
  (\cref{sec:evaluation}).
\end{enumerate}

\noindent
\sys builds upon and references the \acegf identity
framework~\cite{ace-gf}, the \zkace proving system~\cite{zk-ace}, the
AR-ACE attestation relay for post-quantum mempool propagation~\cite{ar-ace},
and the CT-DAP
condition-triggered dormant authorization paths~\cite{ct-dap}.

\section{Background and Motivation}
\label{sec:background}

\subsection{The ACE-GF Cryptographic Primitive}
\label{sec:acegf-primitive}

The Atomic Cryptographic Entity Generative Framework
(\acegf)~\cite{ace-gf} introduces a deterministic, multi-stream key
derivation architecture rooted in a single high-entropy secret called
the Root Entropy Value ($\REV$).  The $\REV$ is encapsulated using
Argon2id~\cite{argon2} with parameters $(m{=}4096,\; t{=}3,\; p{=}1)$
and is never stored in plaintext.  It does not leave the user's client:
the only proof that uses it as a witness, the account-birth proof, is
generated on that client (\cref{sec:attestation}).

\begin{definition}[Multi-Stream Key Derivation]
\label{def:key-derivation}
Given a root entropy value $\REV$ and a purpose-specific information
string $\mathit{info}_i$, the $i$-th key stream is defined as
\begin{equation}
  k_i \;\leftarrow\; \HKDF\text{-}\mathsf{SHA256}(\REV,\;\mathit{info}_i,\;\mathit{salt}_i).
\end{equation}
The reference implementation of \sys uses the following seven
\emph{application-profile} strings as $\mathit{info}_i$.  They are this
paper's profile of the \acegf derivation, not strings defined
in~\cite{ace-gf}, whose $\mathsf{Derive}$ encodes a context tuple
$(\mathsf{AlgID},\mathsf{Domain},\mathsf{Index})$:
\begin{enumerate}[label=\textnormal{Stream \arabic*:},leftmargin=5em]
  \item $\HKDF(\REV, \texttt{"ACEGF-V1-ED25519-SOLANA"})$ for Solana signing,
  \item $\HKDF(\REV, \texttt{"ACEGF-V1-ED25519-POLKADOT"})$ for Polkadot,
  \item $\HKDF(\REV, \texttt{"ACEGF-V1-SECP256K1-EVM"})$ for EVM chains,
  \item $\HKDF(\REV, \texttt{"ACEGF-V1-SECP256K1-BTC"})$ for Bitcoin,
  \item $\HKDF(\REV, \texttt{"ACEGF-V1-SECP256K1-COSMOS"})$ for Cosmos,
  \item $\HKDF(\REV, \texttt{"ACEGF-V1-X25519-IDENTITY"})$ for end-to-end encryption,
  \item $\HKDF(\REV, \texttt{"ACEGF-V1-ML-DSA-44-PQC-IDENTITY"})$ for post-quantum signing.
\end{enumerate}
\end{definition}

\begin{definition}[Context Isolation]
\label{def:context-isolation}
For any two distinct context strings $c_1 \neq c_2$,
\begin{equation}
  \HKDF(\REV,\;\mathit{info},\;\mathit{context}{=}c_1) \;\neq\;
  \HKDF(\REV,\;\mathit{info},\;\mathit{context}{=}c_2)
\end{equation}
except with negligible probability, and the two derived keys are computationally unlinkable under the
pseudorandom function assumption on HKDF.
\end{definition}

\begin{definition}[ZK Identity Commitment]
\label{def:idcom}
The on-chain identity is the ZK-ACE native identifier~\cite{zk-ace},
its relation (C1): the Poseidon2-M31 sponge $\mathsf{S}_1$ of \zkace
applied to the field encoding $\mathsf{Enc}$ of the identity root, the
public protocol constant $\mathsf{s}_{\mathrm{id}}$, and the domain:
\begin{equation}
  \idcom = \mathsf{S}_1\bigl(\mathsf{Enc}(\REV,\;\mathsf{s}_{\mathrm{id}},\;\mathit{domain})\bigr).
\end{equation}
It is not an output of the \acegf HKDF derivation, which produces only
the installed keys.
This 32-byte value is the \emph{only} identity artifact visible on-chain.
It is a one-way image of $\REV$, not a hiding commitment, and every
transaction of an identity carries the same $\idcom$, so an identity's
transactions are linkable.
The keys that authorize transactions are installed separately
(\cref{sec:attestation}), and their public keys are on chain.
\end{definition}

\subsection{Solana TPU Pipeline Analysis}
\label{sec:solana-analysis}

Solana's Transaction Processing Unit processes transactions through a
four-stage pipeline:
\begin{equation*}
  \mathsf{Fetch} \;\to\; \mathsf{SigVerify}\;(\text{GPU}) \;\to\;
  \mathsf{Banking} \;\to\; \mathsf{Broadcast}.
\end{equation*}
The SigVerify stage resides on the critical path and imposes $O(N)$ cost
that scales linearly with the number of transactions.
\Cref{tab:sigverify-cost} presents the per-transaction and per-block
costs of Ed25519 signature verification.

\begin{table}[t]
\centering
\caption{Ed25519 signature verification cost on the Solana pipeline.}
\label{tab:sigverify-cost}
\begin{tabular}{@{}lcc@{}}
\toprule
\textbf{Operation} & \textbf{CPU} & \textbf{GPU (batched)} \\
\midrule
Per-transaction verify & ${\sim}76\;\mu\mathrm{s}$ & ${\sim}2\;\mu\mathrm{s}$ \\
1{,}000-tx block & 76\,ms & 2\,ms \\
10{,}000-tx block & 760\,ms & 20\,ms \\
100{,}000-tx block & 7{,}600\,ms & 200\,ms \\
\bottomrule
\end{tabular}
\end{table}

Two observations describe this pipeline:

\begin{enumerate}[leftmargin=*]
\item \textbf{Serial dependency.}  SigVerify must complete \emph{before}
  execution begins, creating a hard serial dependency on the critical
  path.  Any increase in block size directly increases end-to-end latency.

\item \textbf{Universal GPU requirement.}  Because every validator must
  independently verify every signature, all validators require GPU
  hardware, not only block producers.  This raises the cost of
  participation and centralizes the validator set.
\end{enumerate}

\noindent
\sys keeps native signature verification before execution on every
validator, so both observations also apply to its validators; whether a
GPU is needed depends on the block size (\cref{sec:hardware-cost}).  What
\sys moves off the critical path is verification by parties outside the
validator set.

\subsection{Limitations of Existing ZK-Enhanced Chains}
\label{sec:existing-zk}

Several blockchain systems employ zero-knowledge proofs, but none
applies them to the identity and authorization layer in a way that
achieves $O(1)$ per-block verification:

\begin{itemize}[leftmargin=*]
\item \textbf{zkSync and StarkNet}~\cite{zksync,starknet} use ZK proofs
  for state compression and Layer-2 rollup verification.  Individual
  transaction signatures are still verified in $O(N)$ within the rollup
  sequencer before proof generation.

\item \textbf{Mina Protocol}~\cite{mina} employs recursive SNARK
  composition for constant-size chain state, but per-block signature
  verification remains $O(N)$.

\item \textbf{Aptos and Sui}~\cite{aptos,sui} achieve fast finality via
  DAG-based BFT consensus but retain $O(N)$ per-transaction signature
  verification on the critical path.
\end{itemize}

\sys validators also verify signatures in $O(N)$.  \sys differs in
separating the identifier from the installed key and in offering an
aggregate certificate over public block data to external verifiers.

\section{System Design}
\label{sec:design}

We describe the Attest--Execute--Prove pipeline, attestation
verification, and proof pipelining.

\subsection{Attest--Execute--Prove Pipeline}
\label{sec:pipeline}

\sys processes each slot through two phases that overlap across
consecutive slots.  \textbf{Phase~1} (Attest + Execute) resides on the
critical path and determines the block time; Phase~1a verifies each
transaction's installed-key signature, and the birth proof of each
birth transaction, before execution.  \textbf{Phase~2} (Prove) is
optional: any builder may run it asynchronously, from the published
block and its parent state only, to produce an aggregate finality
certificate.

\begin{algorithm}[t]
\caption{Attest--Execute--Prove Pipeline Scheduler}
\label{alg:pipeline}
\begin{algorithmic}[1]
\Require Transaction pool $\mathcal{T}$, current slot number $S$
\Ensure Published block $B_S$, finality certificate $\mathit{FC}_{S-1}$
\Statex
\Procedure{ProcessSlot}{$S$}
  \State $\mathit{txs} \gets \Call{SelectTransactions}{\mathcal{T},\;\mathit{max}{=}2000}$
  \Statex \Comment{\textbf{Phase 1a: Parallel TxCheck} (header check plus one native signature verification per tx, CPU)}
  \ForAll{$\mathit{tx} \in \mathit{txs}$ \textbf{in parallel}}
    \If{\Call{VerifyAttestation}{$\mathit{tx}.\mathit{attestation},\;\mathit{tx}$} $=$ \textsc{Reject}}
      \State \Call{Reject}{$\mathit{tx}$} \Comment{\Cref{alg:attestation}, against the parent state}
    \EndIf
  \EndFor
  \Statex \Comment{\textbf{Phase 1b: Parallel Execution} (${\sim}10$--$50\;\mu\mathrm{s}$/tx)}
  \ForAll{$\mathit{tx} \in \mathit{valid\_txs}$ via Sealevel-Enhanced scheduler}
    \State $\Delta_{\mathit{tx}} \gets \Call{Execute}{\mathit{tx}}$ \Comment{exact-next nonce checked here}
  \EndFor
  \State $B_S \gets \Call{BuildBlock}{\mathit{valid\_txs},\;\Delta,\;\mathit{attestations}}$
  \State \Call{Broadcast}{$B_S$} \Comment{Soft finality upon $\frac{2}{3}$ votes}
\EndProcedure
\Statex
\Procedure{ProveAsync}{$B_S$} \Comment{\textbf{Phase 2: optional, any builder, public data only}}
  \State $\Pi \gets \emptyset$
  \ForAll{$\mathit{tx} \in B_S.\mathit{transactions}$ \textbf{in parallel}}
    \State $\mathit{pub} \gets$ the public authorization statement of $\mathit{tx}$ \Comment{\cref{sec:proof-pipeline}}
    \State $\pi_{\mathit{tx}} \gets \mathsf{Groth16.Prove}(\mathcal{C}_{\mathsf{auth}},\;\mathit{pub})$ \Comment{no secret witness}
    \State $\Pi \gets \Pi \cup \{\pi_{\mathit{tx}}\}$
  \EndFor
  \While{$|\Pi| > 1$} \Comment{Tree-structured recursive aggregation}
    \State $\Pi \gets \Call{PairwiseAggregate}{\Pi}$
  \EndWhile
  \State $\mathit{FC}_S \gets \bigl(\mathsf{Hash}(B_S),\;S,\;\Pi[0],\;\mathsf{Commit}(B_S.\mathit{pubs})\bigr)$ \Comment{$\mathit{pubs}$: the per-tx vectors $\mathit{pub}$}
  \State \Call{Broadcast}{$\mathit{FC}_S$} \Comment{Hard finality}
\EndProcedure
\end{algorithmic}
\end{algorithm}

\Cref{alg:pipeline} details the complete pipeline logic.
\Cref{tab:critical-path} breaks down the timing of each stage and
identifies whether it lies on the critical path.

\begin{table}[t]
\centering
\caption{Critical path analysis of the Attest--Execute--Prove pipeline.}
\label{tab:critical-path}
\begin{tabular}{@{}lccc@{}}
\toprule
\textbf{Stage} & \textbf{Latency} & \textbf{Critical?} & \textbf{Bottleneck} \\
\midrule
Header check & ${\sim}1$--$5\;\mu\mathrm{s}$/tx & Yes & CPU (lightweight) \\
Sig.\ verify (Ed25519/ML-DSA-44) & ${\sim}76$/${\sim}200\;\mu\mathrm{s}$/tx$^{*}$ & Yes & CPU \\
Birth proof verify & $2.46\;\mathrm{ms}$/birth tx$^{*}$ & Yes & CPU \\
Execute & ${\sim}10$--$50\;\mu\mathrm{s}$/tx & Yes & State I/O \\
Block publish & ${\sim}50\;\mathrm{ms}$ & Yes & Network \\
\midrule
Per-tx proof & not measured & No (async) & Prover \\
Tree aggregation & not measured & No & Prover \\
FC publish & ${\sim}50\;\mathrm{ms}$ & No & Network \\
\bottomrule
\end{tabular}
\\[4pt]
{\footnotesize $^{*}$Single core.  Signature figures are the cited
estimates of \cref{tab:per-tx}; the birth figure is the single-threaded
Circle STARK verification measured in~\cite{zk-ace} on an Apple~M3~Pro.}
\end{table}

The critical path consists solely of the header and signature checks,
Execute, and Block Publish, totaling approximately 300--400\,ms.  At
2{,}000 transactions per block, the signature checks cost
$2{,}000\times 76\;\mu\mathrm{s}=152$\,ms (Ed25519) or
$2{,}000\times 200\;\mu\mathrm{s}=400$\,ms (ML-DSA-44) of single-core
time, that is $9.5$\,ms or $25$\,ms on 16 cores with linear scaling.
Proof generation is fully asynchronous, uses no secret, and no
critical-path step waits for it.

\begin{figure}[t]
\centering
\begin{tikzpicture}[
  >=Stealth,
  slot/.style={draw, minimum height=0.7cm, inner sep=2pt, font=\footnotesize},
  prove/.style={draw, fill=blue!15, minimum height=0.7cm, inner sep=2pt, font=\footnotesize},
  pub/.style={draw, fill=orange!20, minimum height=0.7cm, inner sep=2pt, font=\footnotesize},
  lbl/.style={font=\footnotesize\bfseries, anchor=east},
]
\draw[->, thick] (0, 2.8) -- (13.5, 2.8) node[right]{\footnotesize Time};
\foreach \x/\lab in {0.5/,4.5/,8.5/,12.5/} {
  \draw[dashed, gray] (\x, -1.2) -- (\x, 3.0);
}
\node[font=\footnotesize] at (2.5, 3.15) {Slot $N$ (400\,ms)};
\node[font=\footnotesize] at (6.5, 3.15) {Slot $N{+}1$ (400\,ms)};
\node[font=\footnotesize] at (10.5, 3.15) {Slot $N{+}2$ (400\,ms)};

\node[lbl] at (0.3, 2.0) {CPU:};
\node[slot, fill=green!20, minimum width=3.0cm, anchor=west] at (0.5, 2.0)
  {Attest + Exec $N$};
\node[slot, fill=green!20, minimum width=3.0cm, anchor=west] at (4.5, 2.0)
  {Attest + Exec $N{+}1$};
\node[slot, fill=green!20, minimum width=3.0cm, anchor=west] at (8.5, 2.0)
  {Attest + Exec $N{+}2$};

\node[lbl] at (0.3, 1.1) {Net:};
\node[pub, minimum width=1.2cm, anchor=west] at (3.0, 1.1) {$B_N$};
\node[pub, minimum width=1.2cm, anchor=west] at (7.0, 1.1) {$B_{N+1}$};
\node[pub, minimum width=1.2cm, anchor=west] at (11.0, 1.1) {$B_{N+2}$};

\node[lbl] at (0.3, 0.2) {Prover:};
\node[prove, minimum width=4.2cm, anchor=west] at (3.5, 0.2)
  {Aggregate Block $N$ (optional)};
\node[prove, minimum width=4.2cm, anchor=west] at (7.5, 0.2)
  {Aggregate Block $N{+}1$ (optional)};

\node[pub, minimum width=0.9cm, anchor=west] at (7.4, -0.6) {\scriptsize $\mathit{FC}_N$};
\node[pub, minimum width=0.9cm, anchor=west] at (11.4, -0.6) {\scriptsize $\mathit{FC}_{N+1}$};

\node[font=\scriptsize, red!70!black] at (3.5, -1.15) {soft$_N$};
\node[font=\scriptsize, blue!70!black] at (7.8, -1.15) {hard$_N$};
\draw[<->, thick, red!70!black] (3.5, -1.45) -- (7.8, -1.45);
\node[font=\scriptsize, below] at (5.65, -1.5) {bounded by aggregation time (not measured)};

\end{tikzpicture}
\caption{Proof pipelining across consecutive slots.  CPU-bound
  signature checking and execution occupy the critical path (green),
  while optional aggregate proof generation (blue) and finality
  certificate publication (orange) overlap with subsequent slots.  Box
  widths in the prover row are schematic, not measured durations.}
\label{fig:pipeline}
\end{figure}

\subsection{Attestation Structure and Verification}
\label{sec:attestation}

\paragraph{State.}
Validators keep the \zkace native-account state~\cite{zk-ace} for each
chain domain $d$: the set $\mathcal{C}$ of created identities, the
installed-key table $\mathcal{K}[d,\idcom]$, whose entries are pairs
$(\mathsf{AlgID},\mathsf{pk})$ addressed by a one-byte key index $j$,
and the \zkace registry, operation counter, and consumed replay values.
\sys adds a per-identity transaction sequence number
$\mathcal{Q}[d,\idcom]\in\{0,\ldots,2^{64}-1\}$, which is distinct from
the \zkace operation counter.  The genesis state contains the initial
accounts, which act as the first sponsors.

\paragraph{Account birth.}
An account is created only by the \zkace
\textsc{CreateNativeAccount} transaction~\cite{zk-ace}.  The user's own
client computes $\idcom$ from $\REV$ (\cref{def:idcom}), chooses the
first key $\mathsf{pk}_0$, and proves the \zkace birth relation, which
binds $\idcom$, $\mathsf{pk}_0$, and the sponsor to the proof; $\REV$
and every other witness value stay on the client.  An existing account,
the sponsor, submits the birth payload
$(\idcom,\mathsf{pk}_0,x,\pi,\mathsf{a}_{\mathrm{birth}})$ in a
transaction signed under one of its own installed keys and pays its
fees.  In Phase~1a validators verify the sponsor's signature and run
the \zkace birth checks, including $\mathsf{Verify}(\mathsf{vk},\pi,x)=1$; in
Phase~1b they execute the atomic birth write of~\cite{zk-ace}, which
adds $\idcom$ to $\mathcal{C}$ and sets
$\mathcal{K}[d,\idcom]\gets\{(\mathsf{AlgID}_0,\mathsf{pk}_0)\}$, and
\sys additionally sets $\mathcal{Q}[d,\idcom]\gets 0$.  In the measured
\zkace profile the birth proof is a Circle STARK of $347{,}094$\,bytes
(JSON envelope), proved in about $15$\,ms and verified in about
$2.46$\,ms single-threaded on an Apple~M3~Pro~\cite{zk-ace}.  It is paid
once per account.

\paragraph{Key addition and rotation.}
Adding, replacing, or removing an entry of $\mathcal{K}[d,\idcom]$ is an
ordinary transaction signed under a key already installed for
$\idcom$; a rotation can change the algorithm, for example from Ed25519
to ML-DSA-44, and leaves $\idcom$ unchanged.  No account in
$\mathcal{C}$ is left with an empty key table.  Rotation is
forward-only: it governs which keys authorize future transactions and
does not undo an earlier exposure of an installed secret key or of
$\REV$, from which HKDF-derived keys can be recomputed~\cite{ace-gf}.
Writes to $\mathcal{K}$ take effect from the next block, so every
signature in a block is checked against the parent state.

\paragraph{Signed transactions.}
An ordinary transaction carries a signed header, called its attestation
for continuity with the pipeline names.

\begin{definition}[Attestation]
\label{def:attestation}
An attestation for a transaction with payload $P$ on chain $c$, sent by
$\idcom$ with sequence number $\nu$ under installed key index $j$, is a
tuple $A = (\idcom,\;\mathit{domain},\;\nu,\;j,\;\sigma)$
where $\mathit{obj\_hash}$ is computed from $P$ by the signer and by
every validator and is not carried:
\begin{align}
  \mathit{obj\_hash} &= \mathsf{SHA\text{-}256}(P), \label{eq:objhash}\\
  \mathit{domain} &= c, \label{eq:domain}\\
  \delta &= \mathsf{SHA\text{-}256}\bigl(\texttt{"ACE-RT/TX-V1"} \,\|\,
     \mathit{domain} \,\|\, \idcom \,\|\, \mathsf{LE64}(\nu) \,\|\, j
     \,\|\, \mathit{obj\_hash}\bigr), \label{eq:txdigest}\\
  \sigma &= \mathsf{Sign}_{\mathsf{AlgID}}(\mathsf{sk},\delta),
     \quad (\mathsf{AlgID},\mathsf{pk}) = \mathcal{K}[c,\idcom][j].
\end{align}
The fixed header fields take $49$ bytes (32 bytes $\idcom$, 8 bytes
$\mathit{domain}$, 8 bytes $\nu$, 1 byte $j$).
With the signature the attestation is $49+64=113$ bytes under Ed25519
and $49+2{,}420=2{,}469$ bytes under ML-DSA-44~\cite{nist-pqc}.  The
public key is not carried.
\end{definition}

All fields of~\eqref{eq:txdigest} have fixed length, so the encoding is
injective, and the label separates $\delta$ from the \zkace action
digest and from every other message signed in \sys.  An installed key is
used only for digests of this form.

\begin{algorithm}[t]
\caption{Attestation Generation and Verification}
\label{alg:attestation}
\begin{algorithmic}[1]
\Statex \textbf{Client-side: Generate Attestation}
\Function{GenerateAttestation}{$\mathsf{sk}$, $j$, $\idcom$, $P$, $c$, $\nu$}
  \State $h \gets \mathsf{SHA\text{-}256}(P)$
  \State $\delta \gets$ \eqref{eq:txdigest} on $(c,\idcom,\nu,j,h)$
  \State \Return $(\idcom,\;c,\;\nu,\;j,\;\mathsf{Sign}_{\mathsf{AlgID}}(\mathsf{sk},\delta))$
  \Comment{113\,B (Ed25519) or 2{,}469\,B (ML-DSA-44)}
\EndFunction
\Statex
\Statex \textbf{Validator-side: Verify Attestation (Phase~1a, parent state)}
\Function{VerifyAttestation}{$A$, $\mathit{tx}$}
  \State $h \gets \mathsf{SHA\text{-}256}(\mathit{tx}.\mathit{payload})$ \Comment{Payload hash, computed}
  \If{$A.\mathit{domain} \neq c$ \textbf{or} $A.\idcom \notin \mathcal{C}$ \textbf{or} $\mathcal{K}[c,A.\idcom][A.j]$ undefined}
    \State \Return \textsc{Reject} \Comment{Domain, identity, and key checks}
  \EndIf
  \State $(\mathsf{AlgID},\mathsf{pk}) \gets \mathcal{K}[c,A.\idcom][A.j]$;\quad $\delta \gets$ \eqref{eq:txdigest} on $(c,A.\idcom,A.\nu,A.j,h)$
  \If{$\mathsf{Verify}_{\mathsf{AlgID}}(\mathsf{pk},\delta,A.\sigma) \neq 1$}
    \State \Return \textsc{Reject} \Comment{Native signature check}
  \EndIf
  \If{$\mathit{tx}$ is \textsc{CreateNativeAccount} \textbf{and} the \zkace birth checks fail}
    \State \Return \textsc{Reject} \Comment{Birth proof, \cite{zk-ace}}
  \EndIf
  \State \Return \textsc{Accept}
\EndFunction
\Statex
\Statex \textbf{Execution (Phase~1b):} atomically require $A.\nu = \mathcal{Q}[c,A.\idcom]$ and $A.\nu < 2^{64}-1$, set $\mathcal{Q}[c,A.\idcom] \gets A.\nu+1$, and execute $P$.
\end{algorithmic}
\end{algorithm}

\Cref{alg:attestation} presents the generation and verification
procedures.  The validator performs the header checks at
${\sim}1$--$5\;\mu\mathrm{s}$ per transaction on CPU and one native
signature verification under the installed key's algorithm.  A
transaction without a valid signature under a key in
$\mathcal{K}[c,\idcom]$ is rejected at ingress, so knowing a registered
$\idcom$ is not enough to submit a transaction for it.  The sequence
number is exact-next and is checked against state at execution, so a
replayed or reordered transaction of the same identity fails there;
different identities do not contend on it.

\paragraph{Transaction format comparison.}
\Cref{tab:tx-format} compares the per-transaction on-chain data between
Solana and \sys.

\begin{table}[t]
\centering
\caption{Per-transaction data format comparison.}
\label{tab:tx-format}
\begin{tabular}{@{}lcc@{}}
\toprule
\textbf{Component} & \textbf{Solana} & \textbf{ACE Runtime} \\
\midrule
Header & 3\,B & 3\,B \\
Account keys & $N \times 32$\,B & $N \times 32$\,B \\
Recent blockhash & 32\,B & 32\,B \\
Instructions & variable & variable \\
Authorization & $N \times 64$\,B (signatures) & 113\,B / 2{,}469\,B (attestation$^{*}$) \\
\midrule
Typical total & 250--1{,}232\,B & 243--267\,B / 2{,}599--2{,}623\,B$^{*}$ \\
Public key on chain & In every transaction & Once per installed key \\
\bottomrule
\end{tabular}
\\[4pt]
{\footnotesize $^{*}$Ed25519 / ML-DSA-44 installed key: 49-byte header
plus a 64-byte or 2{,}420-byte signature.}
\end{table}

With an Ed25519 key, a typical \sys transaction (243--267\,B) is about
the size of a typical single-signer Solana transaction (${\sim}250$\,B):
the header carries $\idcom$, $\mathit{domain}$, and $\nu$ in addition to
the signature, and no public key.

\subsection{Proof Pipelining}
\label{sec:proof-pipeline}

The proof generation pipeline is optional and decoupled from the block
production critical path.  Any party with prover hardware (a
\emph{builder}) can run it: its inputs are the published block and the
parent state, both public, and it receives no witness from users.
\Cref{alg:proof-pipeline} details the three-stage proof generation
process.

\begin{algorithm}[t]
\caption{Aggregate Proof Pipeline}
\label{alg:proof-pipeline}
\begin{algorithmic}[1]
\Procedure{ProofPipeline}{}
  \State $T \gets$ prover parallelism
  \While{$\mathit{running}$}
    \State $B \gets \Call{WaitForPublishedBlock}{}$
    \Statex \Comment{\textbf{Stage 1:} Parallel per-tx proofs of the authorization statement $\mathcal{C}_{\mathsf{auth}}$}
    \State $\Pi \gets \Call{ParallelMap}{B.\mathit{txs},\;\mathsf{GenAuthProof},\;T}$
    \Statex \Comment{\textbf{Stage 2:} Tree aggregation ($\log_2 N$ levels)}
    \While{$|\Pi| > 1$}
      \State $\Pi' \gets \emptyset$
      \ForAll{$(i, i{+}1) \in \Pi$ \textbf{in parallel}}
        \State $\Pi' \gets \Pi' \cup \{\mathsf{RecursiveGroth16}(\Pi[i], \Pi[i{+}1])\}$
      \EndFor
      \State $\Pi \gets \Pi'$
    \EndWhile
    \Statex \Comment{\textbf{Stage 3:} Publish finality certificate}
    \State $\Call{PublishFinalityCert}{B.\mathit{slot},\;\Pi[0]}$
    \Comment{$\Pi[0]$: single Groth16 proof, 256\,B}
  \EndWhile
\EndProcedure
\end{algorithmic}
\end{algorithm}

\paragraph{Authorization statement.}
For a transaction at position $i$ of block $B_S$, the public statement
$\mathit{pub}$ consists of $\mathsf{Hash}(B_S)$, the parent state root,
$i$, the attestation of \cref{def:attestation}, and the installed entry
$(\mathsf{AlgID},\mathsf{pk})=\mathcal{K}[c,\idcom][j]$.  The circuit
$\mathcal{C}_{\mathsf{auth}}$ checks that the transaction is the $i$-th
leaf of the block's transaction root, that the key entry opens against
the parent state root, that $\delta$ is computed by~\eqref{eq:txdigest},
and that $\mathsf{Verify}_{\mathsf{AlgID}}(\mathsf{pk},\delta,\sigma)=1$;
for a \textsc{CreateNativeAccount} transaction it also checks the
\zkace birth proof $\mathsf{Verify}(\mathsf{vk},\pi,x)=1$.  Every input,
including the Merkle paths, is public data, so the proof carries no
secret and its zero-knowledge property is not used.  It restates checks
that validators have already performed natively in Phase~1a.

\paragraph{Birth relation.}
The birth proof itself is the \zkace relation~\cite{zk-ace}, proved on
the user's client with public input $x=(\idcom,m,d,n,t,\mathit{rp})$:
\begin{enumerate}[label=(\roman*)]
\item \textbf{Identity binding (C1):}
  $\idcom = \mathsf{S}_1(\mathsf{Enc}(\REV, \mathsf{s}_{\mathrm{id}}, d))$,
  as in \cref{def:idcom}, with $d$ the chain domain.

\item \textbf{Target and replay value (C2)--(C5):}
  $t = \mathsf{S}_t(\mathsf{Enc}(\REV, \mathsf{alg}, \mathsf{ctxdom}, \mathsf{index}))$,
  $\mathit{rp} = \mathsf{S}_r(\mathsf{Enc}(\REV, \mathsf{alg}, \mathsf{ctxdom}, \mathsf{index}, m, d, n))$,
  $\mathsf{ctxdom} = d$, and equality of the full 64-bit public nonce
  $n$ (zero at birth), where $m$ is the \zkace birth digest, which binds
  the sponsor and $\mathsf{pk}_0$ (SHA-256, computed outside the
  circuit).
\end{enumerate}

\noindent
\zkace~\cite{zk-ace} specifies constraints (i)--(ii) as Poseidon2-M31
sponges (eleven permutations, $253$ active AIR rows) and measures a
Circle STARK backend: $347{,}094$\,bytes per proof (JSON envelope), about
$15$\,ms proving and $2.46$\,ms verification single-threaded on an
Apple~M3~Pro.  These are the birth-proof costs of \sys.  \zkace measures
neither a Groth16 instantiation nor recursive aggregation.  The BN254
Groth16~\cite{bn254} instantiation of the aggregate certificate used in
this paper, with a 256-byte proof verifiable in
${\sim}0.5\;\mathrm{ms}$ (three bilinear pairing operations), is
therefore an assumption of this paper, not a measured result.

\paragraph{Aggregation assumption.}
$\mathcal{C}_{\mathsf{auth}}$ verifies an Ed25519 or ML-DSA-44 signature
and, for births, a \zkace proof inside a circuit; neither in-circuit
verifier is specified or measured here, and \zkace notes that
in-circuit signature verification carries the cost of a full signature
checker.  \textsc{RecursiveGroth16} verifies two Groth16 proofs inside a
circuit.  Over BN254 alone this requires non-native pairing arithmetic,
on the order of millions of constraints per verified proof; an
efficient construction needs a pairing-friendly curve cycle or
two-chain, which changes the curve of the final proof and its verifier,
or a folding or accumulation scheme.  None of these is implemented or
measured here, and with recursive Groth16 over BN254 as written in
\cref{alg:proof-pipeline} an aggregation over all
$\lceil\log_2 N\rceil$ levels (eleven for 2{,}000 transactions) within
one slot is not supported.  Hard-finality latency is therefore bounded
by the aggregation time, which this paper does not measure.  Because
validators have already verified every authorization natively, this
latency is a liveness property of hard finality only; it does not
affect soft finality, execution, or the critical path.

\subsection{Two-Tier Finality Model}
\label{sec:finality}

\sys has two finality levels: soft finality from BFT votes over
natively verified authorization, and hard finality from a verified
aggregate certificate.

\begin{definition}[Soft Finality]
\label{def:soft}
Block $B_S$ achieves \emph{soft finality} when $\geq \frac{2}{3}$ of
stake-weighted validator votes confirm the block.  An honest validator
votes only for a block whose every transaction passes Phase~1a against
the parent state and whose re-execution matches the header.  This occurs
at $t_{\mathit{soft}} \approx 400\;\mathrm{ms}$ after slot start.
\end{definition}

\begin{definition}[Hard Finality]
\label{def:hard}
Block $B_S$ achieves \emph{hard finality} when a valid Groth16 finality
certificate $\mathit{FC}_S$ for it is published and verified.  Its
latency after soft finality is bounded by the aggregation time of
\cref{sec:proof-pipeline}, which is not measured.  Hard finality adds a
certificate that parties outside the validator set can check at $O(1)$
cost; it does not add reversal protection for validators beyond soft
finality.
\end{definition}

\begin{algorithm}[t]
\caption{Finality State Machine}
\label{alg:finality}
\begin{algorithmic}[1]
\State \textbf{States:} $\{\textsc{Pending},\;\textsc{Soft},\;\textsc{Hard}\}$
\Statex
\Function{UpdateFinality}{$S$, $\mathit{event}$}
  \Switch{$\mathit{event}$}
    \Case{$\mathsf{BlockReceived}(B)$}
      \State Verify $B.\mathit{header}.\mathit{attest\_merkle\_root}$
      \State Run \Call{VerifyAttestation}{} on every transaction against the parent state
      \State Re-execute $B.\mathit{transactions}$
      \If{valid}
        \State $\Call{Vote}{B.\mathit{slot}}$
      \EndIf
      \If{$\mathit{votes}(B.\mathit{slot}) \geq \frac{2}{3} \cdot \mathit{total\_stake}$}
        \State $B.\mathit{state} \gets \textsc{Soft}$ \Comment{${\sim}400\;\mathrm{ms}$}
      \EndIf
    \EndCase
    \Case{$\mathsf{FinalityCertReceived}(\mathit{FC})$}
      \State \textbf{require} $\mathit{FC}.\mathit{slot} = B.\mathit{slot}$ \textbf{and} $\mathit{FC}.\mathit{block\_hash} = \Call{Hash}{B}$ \textbf{and} $\mathit{FC}.\mathit{pub} = \mathsf{Commit}(B.\mathit{pubs})$
      \If{$B.\mathit{state} = \textsc{Soft} \land \mathsf{Groth16.Verify}(\mathit{FC}.\mathit{proof},\;\mathit{FC}.\mathit{pub},\;\mathit{VK})$}
        \State $B.\mathit{state} \gets \textsc{Hard}$
      \EndIf \Comment{An invalid FC is ignored in every state}
    \EndCase
  \EndSwitch
\EndFunction
\end{algorithmic}
\end{algorithm}

There is no certificate timeout.  A block for which no certificate
arrives stays soft-final: it is never rolled back for that reason, and
no builder is slashed for not producing a certificate.  Validators
therefore do not need to agree on when a certificate arrived.
\Cref{tab:finality-compare} compares finality latencies across major
blockchains.

\begin{table}[t]
\centering
\caption{Finality latency comparison across blockchain platforms.}
\label{tab:finality-compare}
\begin{tabular}{@{}lcccc@{}}
\toprule
\textbf{Chain} & \textbf{Block} & \textbf{Soft} & \textbf{Hard$^{\dagger}$} & \textbf{Rollback Condition} \\
\midrule
Solana & 400\,ms & 400\,ms & ${\sim}12\;\mathrm{s}$ & Fork \\
Ethereum & 12\,s & 12\,s & ${\sim}15\;\mathrm{min}$ & Fork \\
\textbf{ACE Runtime} & \textbf{400\,ms} & \textbf{400\,ms} & soft + aggregation$^{\ddagger}$ & Fork; never a missing FC \\
\bottomrule
\end{tabular}
\begin{flushleft}
\footnotesize $^{\dagger}$Hard finality definitions differ across
systems.  Solana: 31 confirmed slots (optimistic rollback
window)~\cite{solana-whitepaper}.  Ethereum: two Casper FFG epochs
(BFT finality)~\cite{casper-ffg}.  ACE Runtime: valid aggregate
certificate published.  Direct numeric comparison should be interpreted
with these definitional differences in mind.\\
$^{\ddagger}$Aggregation time is not measured
(\cref{sec:proof-pipeline}).
\end{flushleft}
\end{table}

\begin{theorem}[Finality and Authorization]
\label{thm:hard-finality}
Assume Byzantine stake below one third and partial synchrony.  Then
(i) a soft-final block is not reversed, and every transaction in it
passed Phase~1a at some honest validator; and (ii) if $\mathit{FC}_S$
is accepted for $B_S$, every transaction of $B_S$ carries a valid
signature under a key installed for its sender at the parent state, or
is a birth with a valid \zkace proof, unless the adversary breaks the
knowledge soundness of the aggregate Groth16 proof ($q$-PKE on BN254) or
the binding of $\mathsf{Commit}$, or finds a collision in SHA-256 or in
the hash of the transaction and state trees.
\end{theorem}

\begin{proof}
(i) Non-reversal is the assumed safety of the BFT protocol.  At least
$\frac{2}{3}$ of stake voted for $B_S$ and less than $\frac{1}{3}$ is
Byzantine, so more than $\frac{1}{3}$ of stake belongs to honest voters,
each of which voted only after \textsc{VerifyAttestation} accepted every
transaction.  (ii) A validator or light client accepts $\mathit{FC}_S$
only if its public input equals $\mathsf{Commit}(B_S.\mathit{pubs})$ and
$\mathit{FC}_S$ names $\mathsf{Hash}(B_S)$; each $\mathit{pub}$ includes
$\mathsf{Hash}(B_S)$ and the parent state root.  Under the binding of
$\mathsf{Commit}$ an accepting proof attests to the statements of
$B_S$, and under knowledge soundness of the aggregate proof
(\cite{groth16}) each statement of $\mathcal{C}_{\mathsf{auth}}$ holds.
Merkle openings against the roots are bound by collision resistance of
the tree hash, and $\delta$ by that of SHA-256.  Neither part depends on when, or whether, a
certificate arrives.  The certificate attests to authorization, not to
the correctness of execution, which rests on validator re-execution and
consensus.
\end{proof}

\subsection{Block Structure}
\label{sec:block-structure}

\begin{definition}[ACE Block]
An \sys block $B$ consists of:
\begin{itemize}[leftmargin=*]
  \item \textbf{Header} (${\sim}256$ bytes): slot number, parent hash,
    state root, transaction Merkle root, attestation Merkle root, PoH
    hash, leader $\idcom$, timestamp, and transaction count.
  \item \textbf{Body}: a vector of transactions (${\sim}154$ bytes each)
    and their corresponding attestations (113 bytes each with Ed25519
    keys, 2{,}469 bytes with ML-DSA-44 keys).  The body contains one
    signature per transaction and, for birth transactions, the \zkace
    birth proof; the aggregate proof arrives separately in the finality
    certificate.
\end{itemize}
\end{definition}

\begin{definition}[Finality Certificate]
A finality certificate $\mathit{FC}$ is a fixed-size structure consisting
of: block hash (32~bytes), slot number (8~bytes), aggregated Groth16
proof (256~bytes), and a commitment to the list of per-transaction
public statements $\mathit{pub}$ of the block (\cref{alg:pipeline};
32~bytes), which every validator recomputes from the block body and the
parent state: for each transaction it takes $\mathsf{Hash}(B)$, the
parent state root, the position, and the attestation from the block,
and the installed key entry $\mathcal{K}[c,\idcom][j]$ from the parent
state.  The total size is 300--400~bytes, \emph{independent of block
size}.  Verification requires a single Groth16 verification (three bilinear
pairing operations on BN254), completing in ${\sim}0.5\;\mathrm{ms}$
under the assumption of \cref{sec:proof-pipeline}.
\end{definition}

\paragraph{Block size comparison.}
For a block containing 2{,}000 transactions:
\begin{itemize}[leftmargin=*]
  \item \textbf{Solana:} $2{,}000 \times {\sim}1{,}232\;\mathrm{B}
    \approx 2.4\;\mathrm{MB}$ (including signatures and public keys).
  \item \textbf{ACE Runtime, Ed25519 keys:} Header (${\sim}256\;\mathrm{B}$) +
    transactions ($2{,}000 \times {\sim}154\;\mathrm{B} = 308\;\mathrm{KB}$) +
    attestations ($2{,}000 \times 113\;\mathrm{B} = 226\;\mathrm{KB}$)
    $\approx 534\;\mathrm{KB}$, a ${\sim}4.6\times$ reduction relative to
    maximum-size Solana transactions ($2{,}464\;\mathrm{KB}/534\;\mathrm{KB}
    \approx 4.6$; \cref{sec:bandwidth-tps}).
  \item \textbf{ACE Runtime, ML-DSA-44 keys:} $256\;\mathrm{B} +
    308\;\mathrm{KB} + 2{,}000 \times 2{,}469\;\mathrm{B}
    \approx 5.25\;\mathrm{MB}$, against
    $2{,}000 \times 4{,}868\;\mathrm{B} \approx 9.74\;\mathrm{MB}$ for
    maximum-size Solana transactions with ML-DSA-44 signatures and public
    keys (${\sim}1.9\times$).
\end{itemize}

\section{Consensus Design}
\label{sec:consensus}

\sys combines Proof of History (PoH) based leader scheduling with BFT
voting and derives validator keys from \acegf.

\subsection{Verifiable Clock via Proof of History}
\label{sec:poh}

Following the Solana model~\cite{solana-whitepaper}, \sys employs a
sequential SHA-256 hash chain as a verifiable passage-of-time clock:
\begin{equation}
  H_0 \;\to\; H_1 = \mathsf{SHA\text{-}256}(H_0) \;\to\;
  H_2 = \mathsf{SHA\text{-}256}(H_1) \;\to\; \cdots
\end{equation}
PoH serves two purposes: (1) providing deterministic leader scheduling
without requiring consensus on time, and (2) reducing the number of
consensus message rounds by establishing a verifiable ordering of events.

The PoH mechanism is \emph{orthogonal} to the \acegf identity layer and
does not introduce additional cryptographic assumptions.

\subsection{ACE-GF-Native Validator Identity}
\label{sec:validator-identity}

Unlike traditional blockchains where validator keys are independently
generated and managed, \sys derives all validator key material from the
validator's $\REV$ through the \acegf HKDF~\cite{rfc5869} hierarchy:

\begin{align}
  k_{\mathit{validator}} &= \HKDF(\REV,\;\texttt{"ACEGF-V1-VALIDATOR-CONSENSUS"},\;\mathit{validator\_id}), \label{eq:valkey}\\
  k_{\mathit{vote}} &= \HKDF(\REV,\;\texttt{"ACEGF-V1-VALIDATOR-VOTE"},\;\mathit{epoch}). \label{eq:votekey}
\end{align}

This design provides several properties:
\begin{enumerate}[leftmargin=*]
\item \textbf{Unified key management.}  All keys derive from a single
  $\REV$, eliminating the operational burden of managing multiple
  independent key pairs.
\item \textbf{Context isolation.}  Compromise of $k_{\mathit{vote}}$
  does not reveal $k_{\mathit{validator}}$,
  as each is derived with a distinct HKDF info string.
\item \textbf{Epoch-based rotation.}  Vote keys can be rotated every
  epoch without changing the validator's on-chain identity ($\idcom$
  remains constant).
\item \textbf{Post-quantum readiness.}  The ML-DSA-44 stream
  (Stream~7) is available from genesis, so consensus signing can move to
  post-quantum keys when required without a new key setup.
\end{enumerate}

\subsection{Leader Election}
\label{sec:leader-election}

The leader for slot $S$ is selected deterministically using the PoH
output and an epoch-level randomness seed:
\begin{equation}
  \mathsf{Leader}(S) = \mathcal{V}\!\left[
    \mathsf{hash}\bigl(\mathsf{PoH}(S) \;\|\; \mathit{epoch\_seed}\bigr)
    \bmod |\mathcal{V}_w|
  \right],
\end{equation}
where $\mathcal{V}_w$ is the set of validators indexed by stake-weighted
position.  This ensures determinism (all validators compute the same
schedule), unpredictability (dependent on PoH output), and stake-weighted
fairness.

\subsection{BFT Voting and Block Confirmation}
\label{sec:bft-voting}

Block confirmation follows standard BFT voting:
\begin{enumerate}[leftmargin=*]
\item The leader publishes block $B_N$ containing transactions and
  attestations.
\item Each validator verifies the attestation Merkle root, verifies
  every transaction's installed-key signature and every birth proof
  against the parent state, re-executes transactions, and casts a vote
  signed with $k_{\mathit{vote}}$.
\item Upon receiving $\geq \frac{2}{3}$ stake-weighted votes, block
  $B_N$ achieves soft finality.
\item If a finality certificate $\mathit{FC}_N$ arrives, a validator or
  light client verifies a single Groth16 proof
  (${\sim}0.5\;\mathrm{ms}$, CPU-only, assumed), and block $B_N$
  achieves hard finality.
\end{enumerate}

The vote messages are $O(1)$ per validator (standard BFT complexity).
The \acegf identity layer introduces no additional per-vote overhead.

\section{Execution Layer}
\label{sec:execution}

\subsection{Sealevel-Enhanced Parallel Execution}
\label{sec:sealevel}

\sys adopts the Sealevel parallel execution
model~\cite{solana-whitepaper}, in which developers declare the set of
accounts accessed by each transaction.  The runtime constructs a
dependency graph and executes non-conflicting transactions in parallel
while serializing conflicting ones.

\sys extends the base Sealevel model with \acegf context isolation
(\cref{def:context-isolation}).  When a single identity operates
multiple vaults under different contexts (e.g., $\texttt{treasury:0}$
and $\texttt{payroll:0}$), each context derives cryptographically
independent key material.  Provided that the on-chain state addressed
by each context is derived solely from that context's key material
(i.e., no shared global accounts are accessed), the runtime can
\emph{guarantee} that cross-context operations are independent at the
protocol level, eliminating the need for runtime conflict detection in
these cases.

\begin{proposition}[Context-Level Parallelism]
\label{prop:context-parallel}
For any two transactions $\mathit{tx}_1$ and $\mathit{tx}_2$ issued by
the same identity under distinct contexts $c_1 \neq c_2$, the derived
state spaces are cryptographically independent.  If neither transaction
accesses shared global state outside its context-derived address space,
the transactions can be executed in parallel without conflict checking.
\end{proposition}

\begin{proof}
By \cref{def:context-isolation}, $\HKDF(\REV, \mathit{info}, c_1)$ and
$\HKDF(\REV, \mathit{info}, c_2)$ produce computationally unlinkable
keys.  The on-chain accounts derived from these keys occupy disjoint
address spaces.  Therefore, no read--write or write--write conflict can
arise between $\mathit{tx}_1$ and $\mathit{tx}_2$, and parallel
execution is safe.
\end{proof}

\paragraph{Example.}
Consider a user operating a treasury vault (context
$\texttt{treasury:0}$) and a payroll vault (context
$\texttt{payroll:0}$).  On Solana, the developer must explicitly declare
non-overlapping account sets and the runtime verifies this at execution
time.  On \sys, the HKDF-based context isolation provides a
\emph{protocol-level guarantee} of independence, enabling the runtime to
parallelize cross-context transactions without conflict detection,
provided neither transaction touches shared global state (e.g., system
accounts) outside its context-derived address space.

\subsection{EVM Compatibility Layer}
\label{sec:evm}

\sys targets an EVM-compatible execution layer, enabling deployment of
existing Solidity and Vyper smart contracts with minimal or no
modification.  Standard EVM opcodes and gas semantics are supported;
the current prototype does not yet cover all precompiled contracts or
edge-case gas accounting.  In addition, \sys exposes four
\acegf-specific precompiled contracts:

\begin{itemize}[leftmargin=*]
\item \texttt{0x0100}: $\mathsf{id\_com\_verify}(\idcom,\;\pi)
  \to \mathsf{bool}$ --- verify an identity commitment proof.
\item \texttt{0x0101}: $\mathsf{context\_derive}(\mathit{tuple})
  \to \mathsf{address}$ --- derive a context-specific address.
\item \texttt{0x0102}: $\mathsf{admin\_factor\_check}(\mathit{path\_id})
  \to \mathsf{status}$ --- check administrative factor authorization.
\item \texttt{0x0103}: $\mathsf{zkace\_batch\_verify}(\pi,\;
  \mathit{commitments}) \to \mathsf{bool}$ --- batch-verify \zkace proofs.
\end{itemize}

With these precompiles, EVM contracts can call \acegf identity
functions directly.

\section{Security Analysis}
\label{sec:security}

We analyze two primary threat vectors and the quantum security posture
of \sys.

\subsection{Authorization Forgery Resistance}
\label{sec:forgery}

Authorization is checked natively by validators before execution, so
the relevant forgery is a soft-final block that executes a transaction
its sender's key holder did not sign, or that creates an account whose
identity the adversary does not open.

\begin{theorem}[Authorization Unforgeability]
\label{thm:attestation}
Let $\mathsf{Adv}^{\mathrm{forge}}_{\mathcal{A}}$ denote the advantage
of any probabilistic polynomial-time (PPT) adversary $\mathcal{A}$, which
does not hold the secret keys installed for honest identities or their
$\REV$, in causing a soft-final block either (a) to execute a
transaction under an honest $\idcom$ whose digest $\delta$ of
\eqref{eq:txdigest} was not signed by the holder of the indicated
installed key, or (b) to accept a \textsc{CreateNativeAccount} that wins
the \zkace birth-unforgeability game~\cite{zk-ace}.  Under the
assumptions of \cref{thm:hard-finality},
\[
  \mathsf{Adv}^{\mathrm{forge}}_{\mathcal{A}}
  \;\leq\;
  q_{\mathcal{K}}\cdot\mathsf{Adv}^{\mathrm{euf\text{-}cma}}_{\mathsf{Sig}}
  + \mathsf{Adv}^{\mathrm{birth}}_{\zkace} + \epsilon_{\mathrm{CR}},
\]
where $q_{\mathcal{K}}$ bounds the number of honestly installed keys,
$\mathsf{Adv}^{\mathrm{euf\text{-}cma}}_{\mathsf{Sig}}$ is the EUF-CMA
advantage against the installed key's scheme (Ed25519 or ML-DSA-44),
$\mathsf{Adv}^{\mathrm{birth}}_{\zkace}$ is the advantage against
\zkace birth unforgeability, which is bounded in~\cite{zk-ace} and not
re-proved here, and $\epsilon_{\mathrm{CR}}$ is the collision advantage
against SHA-256.  The first and last terms are negligible under the
standard assumptions on the signature scheme and SHA-256.  The birth
term is the \zkace bound, a sum of knowledge-soundness, $\mathsf{S}_1$
collision, and root-recovery advantages; for the measured Stwo profile
\zkace states it under its assumptions without a concrete security
level, and this paper does not claim more.
\end{theorem}

\begin{proof}[Proof sketch]
By \cref{thm:hard-finality}(i), every transaction of a soft-final block
passed \textsc{VerifyAttestation} at an honest validator against the
parent state.  Case (b) is exactly the \zkace birth game, with the
honest validator in the role of the \zkace executor, and contributes
$\mathsf{Adv}^{\mathrm{birth}}_{\zkace}$.  Outside case (b), every entry
of $\mathcal{K}[c,\idcom]$ for an honest $\idcom$ was written either by
its honest birth or by a key operation that itself passed
\textsc{VerifyAttestation} under an earlier entry; by induction over
block order, an adversarial entry first appears only through a forged
signature under an honest key.  A transaction in case (a) therefore
carries a valid $\sigma$ on a $\delta$ that the holder of an honest
$\mathsf{pk}$ did not sign.  If that $\delta$ was never signed by the
holder, the reduction, which guesses the key among $q_{\mathcal{K}}$
and answers other signing queries with its oracle, outputs an EUF-CMA
forgery.  Otherwise $\delta$ was signed for a different payload, which
needs a SHA-256 collision in the payload hash or in $\delta$.  The
label of~\eqref{eq:txdigest}, the fixed-length fields, and the rule
that installed keys sign only such digests exclude reuse of a signature
produced for another message type; $\mathit{domain}$ excludes reuse on
another chain.  Re-submission of the same signed transaction on the
same chain is rejected by the exact-next sequence number at execution,
which is a state-machine property under serializable execution, not a
cryptographic one.  The bound is classical for Ed25519 keys; for the
post-quantum case see \cref{sec:quantum}.  A party that holds $\REV$
can recompute HKDF-derived installed keys and is outside the bound;
rotation does not undo that exposure~\cite{ace-gf}.
\end{proof}

The aggregate certificate plays no role in this theorem.  Its
soundness matters only to parties that rely on it instead of
re-verifying signatures (\cref{thm:hard-finality}(ii)).

\subsection{Missing Proof Attack}
\label{sec:missing-proof}

A builder may fail to publish a certificate, or publish an invalid one.
An invalid certificate fails verification and is ignored.  A missing
certificate leaves the block soft-final: soft finality already rests on
natively verified authorization, so nothing is rolled back, no
transaction is re-queued, and there is no timeout at which validators
must agree that a certificate is late.  Because any party can build the
certificate from public data, another builder can produce it later; the
cost of withholding is delayed hard finality for light clients and
bridges, which is a liveness effect only.  Rewarding certificate
producers is a mechanism-design question (\cref{sec:limitations}).

\subsection{Quantum Security Layering}
\label{sec:quantum}

We consider the quantum security of each of the four layers of \sys:

\begin{enumerate}[leftmargin=*]
\item \textbf{Identity layer.}  The on-chain identity $\idcom =
  \mathsf{S}_1(\mathsf{Enc}(\REV, \mathsf{s}_{\mathrm{id}}, \mathit{domain}))$ is a hash preimage
  problem.  Quantum computers cannot extract
  $\REV$ from $\idcom$ (hash preimage resistance is preserved under
  quantum computation, with Grover providing at most a square-root
  speedup on a 256-bit preimage).  These figures describe the general
  construction with a 256-bit $\REV$; under the 24-word wallet encoding
  profile of~\cite{ace-gf} the root has 128-bit source entropy, which
  gives NIST post-quantum Category~1.

\item \textbf{Transfer authorization layer.}  Ordinary transactions are
  as quantum-secure as the installed key's scheme.  ML-DSA-44
  (FIPS~204~\cite{nist-pqc}) is a post-quantum signature; Ed25519 is not.
  Installed public keys are on chain, so an Ed25519 key is exposed to
  harvest-now-forge-later: an adversary that later runs Shor's algorithm
  can recover its secret key and sign for the account.  The remedy is
  rotation to an ML-DSA-44 key (\cref{sec:attestation}), which keeps
  $\idcom$, before such an adversary exists, since the rotation itself
  is authorized by the key being replaced.  Stream~7 of the profile in
  \cref{def:key-derivation} provides ML-DSA-44 keys from the same $\REV$.

\item \textbf{Birth layer.}  Birth inherits the statement of
  \zkace~\cite{zk-ace}, which assigns no concrete classical or
  post-quantum security level to its Poseidon2-M31 Circle STARK
  profile; this paper claims none either.  A pairing-SNARK birth
  profile would not be post-quantum.

\item \textbf{Certificate layer.}  The aggregate Groth16 proof on BN254
  is \emph{not} quantum-secure, as pairing-based assumptions are broken
  by Shor's algorithm.  A quantum adversary could forge a certificate
  and mislead a light client or bridge that relies on it.  Validators
  are not affected, because they verify authorization natively and
  their finality decisions do not depend on the certificate.  Groth16
  can be replaced by a hash-based proof system for the certificate
  without changes to the identity or transfer layers.
\end{enumerate}

\begin{table}[t]
\centering
\caption{Quantum security posture of each \sys layer.}
\label{tab:quantum}
\begin{tabular}{@{}llll@{}}
\toprule
\textbf{Layer} & \textbf{Primitive} & \textbf{Classical} & \textbf{Quantum} \\
\midrule
Identity & Poseidon2-M31 sponge & $\leq$ REV entropy & $\leq$ half (Grover)$^{\dagger}$ \\
Transfer (Ed25519 key) & Ed25519 & 128-bit & Broken (Shor) \\
Transfer (ML-DSA-44 key) & ML-DSA-44 & Category~2 & Category~2 \\
Birth & \zkace STARK profile & \multicolumn{2}{c}{No concrete level claimed~\cite{zk-ace}} \\
Certificate & Groth16/BN254 & 128-bit & Broken (Shor) \\
\bottomrule
\end{tabular}
\\[4pt]
{\footnotesize $^{\dagger}$256-bit for the general construction; 128-bit
source entropy, NIST Category~1, under the 24-word wallet profile.}
\end{table}

\section{Quantitative Evaluation}
\label{sec:evaluation}

We present a quantitative analysis combining representative library-level
microbenchmarks with a simple analytical model.  Unless explicitly marked
as measured prototype numbers, values in this section should be interpreted
as model-driven estimates.  \Cref{sec:prototype-measurements} provides
measured numbers from a Rust reference implementation that validate
the analytical estimates.

\subsection{Per-Transaction Verification Cost}
\label{sec:per-tx-cost}

\Cref{tab:per-tx} compares the per-transaction authorization verification
cost across different cryptographic schemes.

\begin{table}[t]
\centering
\caption{Per-transaction authorization verification cost (representative microbenchmarks and estimates).}
\label{tab:per-tx}
\begin{tabular}{@{}lccc@{}}
\toprule
\textbf{Operation} & \textbf{CPU} & \textbf{GPU (batch)} & \textbf{Source} \\
\midrule
Ed25519 verify (Solana) & ${\sim}76\;\mu\mathrm{s}$ & ${\sim}2\;\mu\mathrm{s}$ & \cite{dalek-ed25519} \\
ECDSA secp256k1 verify & ${\sim}50\;\mu\mathrm{s}$ & ${\sim}3\;\mu\mathrm{s}$ & \cite{libsecp256k1} \\
ML-DSA-44 verify (PQC) & ${\sim}200\;\mu\mathrm{s}$ & N/A & \cite{liboqs} (approx.) \\
Header check (ACE, excl.\ signature) & ${\sim}1$--$5\;\mu\mathrm{s}$ & N/A (CPU sufficient) & \cref{sec:prototype-measurements} \\
Groth16 verify (per block) & ${\sim}500\;\mu\mathrm{s}$ & N/A & \cite{arkworks} \\
\bottomrule
\end{tabular}
\end{table}

A \sys validator performs the header check and one signature
verification under the installed key's scheme, so its per-transaction
cost is that of the Ed25519 or ML-DSA-44 row plus the header check, and
the signature dominates.  Birth transactions add one \zkace proof
verification (${\sim}2.46$\,ms single-threaded in the measured STARK
profile~\cite{zk-ace}).  The Groth16 row is paid only by a party that
verifies the aggregate certificate.

\subsection{Block Verification Time: \texorpdfstring{$O(N)$}{O(N)} vs.\ \texorpdfstring{$O(1)$}{O(1)}}
\label{sec:block-verification}

A \sys validator verifies a block's authorization in $O(N)$ time, as a
Solana validator does.  The difference is for a verifier outside the
validator set, such as a light client or bridge, which can check every
authorization in the block through the aggregate certificate in $O(1)$
time instead of re-verifying $N$ signatures.  \Cref{tab:block-verify}
presents verification times for such a verifier as a function of block
size, under the Groth16 assumption of \cref{sec:proof-pipeline}.

\begin{table}[t]
\centering
\caption{Model-based block verification time for an external verifier:
  re-verifying $N$ Ed25519 signatures versus one aggregate certificate.
  Validators always pay the $O(N)$ column.}
\label{tab:block-verify}
\begin{tabular}{@{}rcccc@{}}
\toprule
\textbf{Block size} & \textbf{$N$ sigs (GPU)} & \textbf{$N$ sigs (CPU)} & \textbf{ACE FC (1 Groth16)} & \textbf{Speedup} \\
\midrule
1{,}000 tx & 2\,ms & 76\,ms & \textbf{0.5\,ms} & $4\times$ \\
10{,}000 tx & 20\,ms & 760\,ms & \textbf{0.5\,ms} & $40\times$ \\
100{,}000 tx & 200\,ms & 7{,}600\,ms & \textbf{0.5\,ms} & $400\times$ \\
1{,}000{,}000 tx & 2{,}000\,ms & 76{,}000\,ms & \textbf{0.5\,ms} & $4{,}000\times$ \\
\bottomrule
\end{tabular}
\end{table}

\begin{figure}[t]
\centering
\begin{tikzpicture}
\begin{axis}[
  width=0.85\columnwidth,
  height=6cm,
  xlabel={Transactions per block ($N$)},
  ylabel={Verification time (ms)},
  xmode=log,
  ymode=log,
  log basis x=10,
  log basis y=10,
  xmin=500, xmax=2000000,
  ymin=0.3, ymax=100000,
  legend style={at={(0.02,0.98)}, anchor=north west, font=\small},
  grid=both,
  grid style={gray!30},
  tick label style={font=\small},
  label style={font=\small},
]
\addplot[blue, thick, mark=square*, mark size=2pt] coordinates {
  (1000, 2) (10000, 20) (100000, 200) (1000000, 2000)
};
\addlegendentry{$N$ signatures (GPU)}

\addplot[red, thick, mark=triangle*, mark size=2pt] coordinates {
  (1000, 76) (10000, 760) (100000, 7600) (1000000, 76000)
};
\addlegendentry{$N$ signatures (CPU)}

\addplot[green!60!black, ultra thick, mark=diamond*, mark size=2.5pt] coordinates {
  (1000, 0.5) (10000, 0.5) (100000, 0.5) (1000000, 0.5)
};
\addlegendentry{ACE certificate $O(1)$}

\end{axis}
\end{tikzpicture}
\caption{Block verification time for an external verifier as a
  function of block size.  Re-verifying signatures scales linearly
  ($O(N)$), while checking the \sys certificate stays at the assumed
  $0.5\;\mathrm{ms}$ ($O(1)$).  \sys validators verify signatures and
  follow the $O(N)$ curves.}
\label{fig:verification-scaling}
\end{figure}

The verification complexity can be summarized as:
\begin{align}
  T_{\mathrm{Solana}}(N) &= O(N) \cdot t_{\mathrm{sig}},
    \quad\text{linear growth with block size}, \label{eq:solana-verify}\\
  T_{\mathrm{ACE}}^{\mathrm{FC}}(N) &= O(1) \cdot t_{\mathrm{Groth16}} = 0.5\;\mathrm{ms},
    \quad\text{constant, external verifier}, \label{eq:ace-verify}\\
  T_{\mathrm{ACE}}^{\mathrm{val}}(N) &= O(N) \cdot t_{\mathrm{sig}},
    \quad\text{validator}. \nonumber
\end{align}

\subsection{On-Chain Authorization Data}
\label{sec:onchain-data}

\Cref{tab:onchain-data} compares the per-block authorization data
footprint, including the post-quantum case.

\begin{table}[t]
\centering
\caption{Per-block authorization data in the block body (signatures,
  public keys, and the \sys attestation header).  The optional 328-byte
  certificate is published separately and is not included.}
\label{tab:onchain-data}
\begin{tabular}{@{}rccccc@{}}
\toprule
\textbf{Block} & \textbf{Solana} & \textbf{Solana} & \textbf{ACE} & \textbf{ACE} & \textbf{ACE PQC vs.} \\
\textbf{size} & \textbf{(Ed25519)} & \textbf{(ML-DSA-44)} & \textbf{(Ed25519)} & \textbf{(ML-DSA-44)} & \textbf{Solana PQC} \\
\midrule
1K tx & 96\,KB & 3{,}732\,KB & 113\,KB & 2{,}469\,KB & $1.5\times$ smaller \\
10K tx & 960\,KB & 37{,}320\,KB & 1{,}130\,KB & 24{,}690\,KB & $1.5\times$ smaller \\
\bottomrule
\end{tabular}
\end{table}

Per transaction, Solana carries 96\,bytes (Ed25519 signature and public
key) or 3{,}732\,bytes (ML-DSA-44), while \sys carries its attestation,
113 or 2{,}469\,bytes, and no public key, because public keys are
installed once in $\mathcal{K}$.  In the post-quantum case this saves
the 1{,}312-byte ML-DSA-44 public key per transaction, a factor of
$3{,}732/2{,}469\approx 1.5$.  In the Ed25519 case \sys carries about
$18\%$ more ($113/96\approx 1.18$): its signature and key index take
65\,bytes, less than Solana's 96, but its header also carries $\idcom$,
$\mathit{domain}$, and $\nu$ (48\,bytes), which have no counterpart in
Solana's 96 bytes.  The aggregate certificate does not
shrink this data; it changes what an external verifier must check, not
what the block carries.

\subsection{Bandwidth-Limited TPS}
\label{sec:bandwidth-tps}

We analyze the theoretical throughput ceiling imposed by network
bandwidth:

\begin{align}
  \text{Solana per-tx:}\quad &{\sim}154\;\mathrm{B}\;(\text{payload}) + 64\;\mathrm{B}\;(\text{sig}) \nonumber\\
  &\quad + 32\;\mathrm{B}\;(\text{pubkey}) + {\sim}982\;\mathrm{B}\;(\text{overhead})
  \approx 1{,}232\;\mathrm{B}, \nonumber\\
  \text{ACE per-tx (Ed25519):}\quad &{\sim}154\;\mathrm{B}\;(\text{payload}) + 113\;\mathrm{B}\;(\text{attestation}) = 267\;\mathrm{B}, \nonumber\\
  \text{ACE per-tx (ML-DSA-44):}\quad &{\sim}154\;\mathrm{B}\;(\text{payload}) + 2{,}469\;\mathrm{B}\;(\text{attestation}) = 2{,}623\;\mathrm{B}. \nonumber
\end{align}

At 1\,Gbps network bandwidth:
\begin{align}
  \mathrm{TPS}_{\mathrm{Solana}} &= \frac{125\;\mathrm{MB/s}}{1{,}232\;\mathrm{B}} \approx 101{,}000\;\text{tx/s}, \label{eq:tps-solana}\\
  \mathrm{TPS}_{\mathrm{ACE}} &= \frac{125\;\mathrm{MB/s}}{267\;\mathrm{B}} \approx 468{,}000\;\text{tx/s}\ \text{(Ed25519)}, \label{eq:tps-ace}\\
  \mathrm{TPS}_{\mathrm{ACE}}^{\mathrm{PQ}} &= \frac{125\;\mathrm{MB/s}}{2{,}623\;\mathrm{B}} \approx 47{,}700\;\text{tx/s}\ \text{(ML-DSA-44)}. \nonumber
\end{align}

In Ed25519 mode \sys has a ${\sim}4.6\times$ ($1{,}232/267\approx 4.6$)
higher bandwidth-limited ceiling \emph{for block propagation} than
Solana's 1{,}232-byte maximum transaction size.  For a typical
single-signer Solana transaction (${\sim}250$\,B, \cref{tab:tx-format})
\sys is about $7\%$ larger (267\,B against 250\,B), so the ${\sim}4.6\times$ figure is
an upper bound reached only for maximum-size, multi-signer transactions,
and there is no bandwidth advantage for typical transactions.  In
ML-DSA-44 mode the comparison is against a maximum-size Solana
transaction that carries an ML-DSA-44 signature and public key,
$1{,}232-96+3{,}732=4{,}868$\,bytes and
$125\;\mathrm{MB/s}/4{,}868\;\mathrm{B}\approx 25{,}700$\,tx/s, so \sys
is ${\sim}1.9\times$ ($4{,}868/2{,}623\approx 1.86$) higher, the
difference being the public key that \sys does not carry.  There is no
witness traffic: provers use only public block data.

\paragraph{CPU-bound TPS.}  We take a 16-core validator, the 400\,ms
slot less ${\sim}50$\,ms for block publication (350\,ms of compute), and
the per-transaction single-core costs of \cref{tab:critical-path}:
header check 1--5\,$\mu$s, signature 76\,$\mu$s (Ed25519) or
200\,$\mu$s (ML-DSA-44), execution 10--50\,$\mu$s.  Header and
signature checks are independent across transactions and spread over
the 16 cores.  If every transaction conflicts, execution is serial and
costs 10--50\,$\mu$s of wall-clock time per transaction, so one
transaction takes $76/16+5/16+50\approx 55\;\mu$s to
$76/16+1/16+10\approx 15\;\mu$s with Ed25519 and $63$ to $23\;\mu$s
with ML-DSA-44.  This gives about 6{,}400--23{,}600 (Ed25519) or
5{,}600--15{,}500 (ML-DSA-44) transactions per slot, that is
${\sim}16{,}000$--$59{,}000$ and ${\sim}14{,}000$--$39{,}000$ TPS.  If
execution parallelizes across all 16 cores (no conflicts), the whole
per-transaction cost is divided by 16: $(76+5+50)/16\approx 8.2\;\mu$s
to $(76+1+10)/16\approx 5.4\;\mu$s with Ed25519, about 107{,}000--161{,}000
TPS, and $15.9$ to $13.2\;\mu$s with ML-DSA-44, about
55{,}000--66{,}000 TPS, above the $47{,}700$\,tx/s bandwidth ceiling of
\cref{eq:tps-ace}, which then binds.  Birth transactions add
$2.46$\,ms of single-core verification each and are excluded.  These
are projections from cited per-operation costs, not measurements of a
running network; the measured prototype header check and mock execution
of \cref{sec:prototype-measurements} are below these estimates.

\subsection{Hardware Cost Comparison}
\label{sec:hardware-cost}

\Cref{tab:hardware} compares the hardware requirements and costs
for different node roles.

\begin{table}[t]
\centering
\caption{Hardware requirements and estimated costs by role.}
\label{tab:hardware}
\begin{tabular}{@{}lp{3cm}p{3cm}p{3cm}@{}}
\toprule
\textbf{Component} & \textbf{Solana Validator} & \textbf{ACE Validator (non-builder)} & \textbf{ACE Builder (optional)} \\
\midrule
CPU & 24+ cores (${\sim}$\$2{,}000) & 16+ cores (${\sim}$\$1{,}500) & 24+ cores (${\sim}$\$2{,}000) \\
RAM & 512\,GB+ (${\sim}$\$3{,}000) & 256--512\,GB (${\sim}$\$1{,}500--3{,}000) & 512\,GB+ (${\sim}$\$3{,}000) \\
GPU & \textbf{Required} (${\sim}$\$2{,}000) & \textbf{Not needed (\$0)} & Required (${\sim}$\$2{,}000) \\
\midrule
Total & ${\sim}$\$7{,}500+ & ${\sim}$\$3{,}500--5{,}000 & ${\sim}$\$7{,}500 \\
Monthly ops & \$1{,}000--3{,}000 & \$600--2{,}000 & \$1{,}000--3{,}000 \\
\bottomrule
\end{tabular}
\end{table}

A non-builder \sys validator needs no GPU at the 2{,}000-transaction
block target: the signature checks take 9.5\,ms (Ed25519) or 25\,ms
(ML-DSA-44) on 16 cores (\cref{sec:pipeline}).  At 16{,}000 TPS
they occupy $16{,}000\times 76\;\mu\mathrm{s}\approx 1.2$ cores with
Ed25519 and $16{,}000\times 200\;\mu\mathrm{s}=3.2$ cores with
ML-DSA-44; at the parallel-execution ceilings of
\cref{sec:bandwidth-tps} they occupy about 8 and 10 cores.  This follows
from the load, not from the authorization design: a Solana validator
verifying the same number of Ed25519 signatures has the same CPU budget.
The hardware cost difference in \cref{tab:hardware} is therefore not
attributable to \sys's authorization path.  The builder role is
optional and needs prover hardware only to produce certificates.

\subsection{Reference Implementation Measurements}
\label{sec:prototype-measurements}

To validate the analytical models above, we implemented a Rust reference
prototype of the \sys pipeline and measured its performance using the
Criterion statistical benchmarking framework~\cite{criterion} on an
Apple~M3~Pro CPU (12 cores, 36\,GB RAM).  Unless otherwise noted, all reported latencies are the sample
median (p50) from Criterion's default configuration (100 samples,
5\,s measurement window per benchmark).
Phases~1a and~1b use Rayon work-stealing parallelism; the current
mock Phase~2 implementation is sequential.  Single-operation
benchmarks are single-threaded.  The benchmark harnesses
(\texttt{benches/\allowbreak crypto\_bench.rs},
\texttt{benches/\allowbreak pipeline\_bench.rs})
are included in the reference implementation and can be
reproduced via \texttt{cargo bench}.

\paragraph{Code availability.}
The reference implementation is publicly available at
\url{https://github.com/ya-xyz/ace-chain/tree/master/ace-runtime}.
For reproducibility, the measurements in this section correspond to the
following pinned revision:
\url{https://github.com/ya-xyz/ace-chain/tree/b805ea90a8c3ab1ce5fa62cf842dccfb9452fc9f/ace-runtime}.
In addition to \texttt{cargo bench}, the end-to-end prototype runner can
be executed via \texttt{cargo run --release --bin prototype\_bench -- 50}.

\paragraph{Cryptographic microbenchmarks.}
\Cref{tab:measured-crypto} summarizes the measured per-operation latencies
for the core cryptographic primitives.

\begin{table}[t]
\centering
\caption{Measured cryptographic primitive latencies (Apple~M3~Pro,
  single-threaded, sample median over 100 samples unless noted).}
\label{tab:measured-crypto}
\begin{tabular}{@{}lr@{}}
\toprule
\textbf{Operation} & \textbf{Latency} \\
\midrule
Payload binding check (SHA-256) & $334\;\mathrm{ns}$ \\
HKDF single-key derivation & $1.18\;\mu\mathrm{s}$ \\
HKDF all 7 profile streams & $7.87\;\mu\mathrm{s}$ \\
Identity commitment, \zkace (C1)$^\dagger$ & $3.36\;\mu\mathrm{s}$ \\
SHA-256 hash (244\,B input) & $544\;\mathrm{ns}$ \\
REV encapsulation (Argon2id)$^*$ & $3.57\;\mathrm{ms}$ \\
\bottomrule
\end{tabular}
\\[4pt]
{\footnotesize $^*$\,10 samples (Argon2id is intentionally slow).
$^\dagger$\,Native $\mathsf{S}_1(\mathsf{Enc}(\REV,\mathsf{s}_{\mathrm{id}},d))$,
three Poseidon2-M31 permutations, measured in the \zkace reference
crate~\cite{zk-ace} at commit \texttt{00d13f9}, single-threaded,
Criterion default configuration; the \sys prototype receives
$\idcom$ and does not compute it from $\REV$.}
\end{table}

The prototype's Phase~1a AttestCheck measured the earlier header
format, which carried $\mathit{obj\_hash}$: it compared the payload
hash with the carried value, matched the domain, and looked up the
identity registry, yielding ${\sim}129\;\mathrm{ns}$/tx amortised at
batch size 2{,}000 (\cref{tab:measured-pipeline}), below the
$1$--$5\;\mu\mathrm{s}$ estimate in \cref{tab:per-tx}.  The header of
\cref{def:attestation} is not measured.  The measured prototype does not verify
an installed-key signature or a birth proof in Phase~1a; these
measurements therefore cover the header check only, and the full
Phase~1a cost is dominated by the signature verification of
\cref{tab:per-tx}, which is not measured here.

\paragraph{Pipeline phase timings.}
\Cref{tab:measured-pipeline} reports end-to-end pipeline phase timings
at batch sizes of 100, 500, and 2{,}000 transactions.  Phases~1a and~1b
execute in parallel via Rayon; Phase~2 uses a mock prover that performs
tree aggregation with hash-based proofs (real Groth16 proving would
dominate the Phase~2 cost on GPU hardware).

\begin{table}[t]
\centering
\caption{Measured pipeline phase timings (Rayon parallelism, Apple~M3~Pro).
  Phase~2 uses a mock prover; real Groth16 proving is GPU-bound and not
  included in these CPU measurements.}
\label{tab:measured-pipeline}
\begin{tabular}{@{}lrrrr@{}}
\toprule
\textbf{Phase} & \textbf{100\,tx} & \textbf{500\,tx} & \textbf{2{,}000\,tx} & \textbf{per-tx\,@\,2K} \\
\midrule
1a: AttestCheck$^\dagger$ & $60.6\;\mu\mathrm{s}$ & $117.8\;\mu\mathrm{s}$ & $258.3\;\mu\mathrm{s}$ & $129\;\mathrm{ns}$ \\
1b: Execute$^\dagger$ & $55.1\;\mu\mathrm{s}$ & $81.6\;\mu\mathrm{s}$ & $141.9\;\mu\mathrm{s}$ & $71\;\mathrm{ns}$ \\
1b: State root (Merkle) & $28.2\;\mu\mathrm{s}$ & $135.1\;\mu\mathrm{s}$ & $538.4\;\mu\mathrm{s}$ & $269\;\mathrm{ns}$ \\
2: Prove (mock) & $214.3\;\mu\mathrm{s}$ & $1.09\;\mathrm{ms}$ & $4.19\;\mathrm{ms}$ & $2.10\;\mu\mathrm{s}$ \\
Block build & $101.6\;\mu\mathrm{s}$ & $515.9\;\mu\mathrm{s}$ & $2.02\;\mathrm{ms}$ & $1.01\;\mu\mathrm{s}$ \\
\midrule
\textbf{End-to-end} & $\mathbf{521\;\mu\mathrm{s}}$ & $\mathbf{1.95\;\mathrm{ms}}$ & $\mathbf{7.69\;\mathrm{ms}}$ & $\mathbf{3.85\;\mu\mathrm{s}}$ \\
\bottomrule
\end{tabular}
\\[4pt]
{\footnotesize $^\dagger$\,Rayon work-stealing parallelism.}
\end{table}

At 2{,}000 transactions, the per-transaction amortised cost of Phase~1a
(AttestCheck) is ${\sim}129\;\mathrm{ns}$, indicating strong parallel
scaling on this prototype (a single-threaded baseline comparison is
left to future work).  The measured CPU-side pipeline (excluding
signature verification and Groth16 proving) completes in under
$8\;\mathrm{ms}$.  Adding the cited single-core signature costs with
linear scaling over the 12 cores gives
$2{,}000\times 76\;\mu\mathrm{s}/12\approx 12.7$\,ms (Ed25519) or
$2{,}000\times 200\;\mu\mathrm{s}/12\approx 33$\,ms (ML-DSA-44), an
estimate that remains well inside the $400\;\mathrm{ms}$ slot budget
assumed in \cref{sec:consensus}.

\paragraph{Implications for production throughput.}
The measured end-to-end CPU cost of $3.85\;\mu\mathrm{s}$ per
transaction (2{,}000\,tx batch) excludes signature verification, so it
does not bound the throughput of the full Phase~1a.  Under the same
linear-scaling estimate, signature checks add about $76/12\approx
6.3\;\mu\mathrm{s}$ (Ed25519) or $200/12\approx 16.7\;\mu\mathrm{s}$
(ML-DSA-44) per transaction, and they, not proving, are the CPU
critical-path cost.  Proving is off the critical path and optional
(\cref{sec:proof-pipeline}); its cost determines only how soon hard
finality follows, and measuring it requires a production aggregation
backend, which we leave to future work.

\section{Related Work}
\label{sec:related}

\Cref{tab:positioning} compares \sys with existing blockchain
platforms.

\begin{table*}[t]
\centering
\caption{Positioning matrix: \sys compared with existing blockchain
  platforms.}
\label{tab:positioning}
\resizebox{\textwidth}{!}{%
\begin{tabular}{@{}lcccccc@{}}
\toprule
\textbf{Platform} & \textbf{Sig.\ Verify} & \textbf{Block Verify} & \textbf{Hard Finality} & \textbf{PQC Path} & \textbf{GPU Req.} & \textbf{Layer} \\
\midrule
Solana~\cite{solana-whitepaper} & $O(N)$ GPU & $O(N)$ & ${\sim}12\;\mathrm{s}$ & None & All nodes & L1 \\
Aptos/Sui~\cite{aptos,sui} & $O(N)$ CPU & $O(N)$ & ${\sim}2\;\mathrm{s}$ & None & None & L1 \\
Ethereum~\cite{ethereum} & $O(N)$ CPU & $O(N)$ & ${\sim}15\;\mathrm{min}$ & EIP planned & None & L1 \\
zkSync~\cite{zksync} & $O(N) \to O(1)$ & $O(1)$ & Inherits L1 & Inherits L1 & Prover & L2 \\
StarkNet~\cite{starknet} & $O(N) \to O(1)$ & $O(1)$ & Inherits L1 & STARK-native & Prover & L2 \\
Mina~\cite{mina} & $O(N)$ & $O(1)^*$ & ${\sim}15\;\mathrm{min}$ & None & Recursive & L1 \\
\textbf{ACE Runtime} & $O(N)$ CPU & $O(N)$; $\mathbf{O(1)}$ via FC$^{\ddagger}$ & soft + aggregation$^{\S}$ & \textbf{ML-DSA-44 installed keys$^\dagger$} & \textbf{Optional prover} & \textbf{L1} \\
\bottomrule
\end{tabular}%
}
\begin{flushleft}
\footnotesize $^*$Mina's $O(1)$ refers to chain-level state proof, not
per-block authorization verification.
$^\ddagger$Validators verify $O(N)$ signatures; the $O(1)$ certificate
check is for light clients and bridges.
$^\S$Aggregation time is not measured (\cref{sec:proof-pipeline}).
$^\dagger$Transfers signed under ML-DSA-44 keys are post-quantum at the
signature level, and an account moves to them by rotation without
changing its identifier; Ed25519 keys are not post-quantum, births
inherit the \zkace statement, and the Groth16 certificate is not
post-quantum (\cref{sec:quantum}).
\end{flushleft}
\end{table*}

\paragraph{High-performance L1 blockchains.}
Solana~\cite{solana-whitepaper} pioneered the PoH-based high-throughput
architecture with 400\,ms block times.  However, its hard finality
requires 31 confirmations (${\sim}12\;\mathrm{s}$), and every validator
must operate GPU hardware for SigVerify.  Aptos~\cite{aptos} and
Sui~\cite{sui} achieve faster finality through DAG-based BFT consensus
but retain $O(N)$ per-transaction verification.  \sys also retains
$O(N)$ native verification on validators at Solana's block time; it
separates the identifier from the installed key and adds an optional
certificate for external verifiers.

\paragraph{ZK rollups and L2 solutions.}
zkSync~\cite{zksync} and StarkNet~\cite{starknet} apply zero-knowledge
proofs for state compression and rollup verification, achieving $O(1)$
state verification for Layer-2 transactions.  However, within the rollup
sequencer, individual transaction signatures are still verified in
$O(N)$, and finality depends on the underlying Layer-1 settlement.  \sys
is a sovereign Layer-1 chain that uses ZK proofs at account birth and in
an off-critical-path certificate that gives external verifiers $O(1)$
verification; its validators verify signatures in $O(N)$, as a rollup
sequencer does.

\paragraph{Recursive proof chains.}
Mina Protocol~\cite{mina} uses recursive SNARK composition to maintain
a constant-size chain state proof.  This addresses the chain
synchronization problem but does not reduce per-block authorization
verification cost, which remains $O(N)$.  The \sys certificate gives
an external verifier block-level $O(1)$ authorization verification,
which complements chain-level compression; validators still pay
$O(N)$.

\paragraph{Identity-layer innovations.}
Decentralized identity (DID) frameworks and self-sovereign identity
(SSI) protocols~\cite{did-spec} have explored alternative identity
models.  \sys is built on the \acegf framework~\cite{ace-gf} and
keeps per-transaction signatures; what it integrates into the execution
layer is an identifier that is independent of the signature scheme.

\section{Limitations and Future Work}
\label{sec:limitations}

\paragraph{Execution-bound throughput.}
Without state sharding, conflicting transactions execute serially,
and the execution-bound throughput of \sys is then about
16{,}000--59{,}000 TPS with Ed25519 keys under the model of
\cref{sec:bandwidth-tps}, the same bound a Solana validator meets for
the same workload.  The authorization layer optimization removes one bottleneck
but does not inherently scale the state I/O layer.  Context-based state
sharding, leveraging the HKDF context isolation property, is the subject
of a companion paper on ACE-GF-native state
partitioning~\cite{ace-context-sharding}.

\paragraph{Linear validator verification.}
Every validator verifies one installed-key signature per transaction
before execution, so validator verification is $O(N)$, as on Solana, and
ML-DSA-44 keys cost about $200\;\mu\mathrm{s}$ per transaction against
${\sim}76\;\mu\mathrm{s}$ for Ed25519 (cited estimates, not measured
here).  The prototype measured in \cref{sec:prototype-measurements}
does not yet perform these checks or the birth-proof check.

\paragraph{Groth16 quantum vulnerability.}
The Groth16 proof system on BN254 is not quantum-secure.  An adversary
able to break Groth16 soundness could forge a certificate and mislead a
light client or bridge that relies on it; validators verify
authorization natively and are not affected.  A hash-based certificate
proof system would remove this gap for external verifiers.  The PQC
migration path is detailed in a companion paper~\cite{ace-pqc-paper}.

\paragraph{Ecosystem bootstrap.}
As a new Layer-1 platform, \sys requires building developer tooling,
wallet infrastructure, and application ecosystem from scratch.  The EVM
compatibility layer (\cref{sec:evm}) mitigates this by enabling
deployment of existing Solidity contracts, but \acegf-native
applications require new development patterns.

\paragraph{Prover market design.}
Any party can produce a certificate from public data, but the design
does not yet specify how certificate producers are rewarded.  A
decentralized prover marketplace, in which multiple competing provers
bid to generate finality certificates, would improve the liveness of
hard finality and reduce centralization risk.  We leave the
mechanism design of such a marketplace to future work.

\paragraph{Trusted setup.}
Groth16 requires a circuit-specific trusted setup ceremony.  While
universal setup alternatives (e.g., PLONK~\cite{plonk}) exist, they
incur larger proof sizes or higher verification costs.  We plan to
evaluate the trade-off between setup trust assumptions and verification
efficiency in a future engineering report.

\section{Conclusion}
\label{sec:conclusion}

This paper has presented \sys, a blockchain runtime that exploits
identity--authorization separation to restructure the transaction
processing pipeline.  Accounts are born through the \zkace birth
transaction, proved on the user's client; ordinary transactions carry
one signature under an installed key, verified natively before
execution and voting; and an optional aggregate certificate over public
block data serves external verifiers.  Supported by analytical modeling
and formal analysis, \sys provides:

\begin{enumerate}[leftmargin=*]
\item \textbf{Authorized soft finality without rollback machinery.}
  Soft finality at ${\sim}400\;\mathrm{ms}$ certifies natively verified
  authorization.  A missing certificate delays hard finality only;
  there is no rollback, timeout, or timeout slashing, and provers hold
  no user secret.  Hard-finality latency is bounded by an aggregation
  time that is not measured here.

\item \textbf{$O(1)$ block verification for external verifiers.}
  Light clients and bridges can check every authorization of a block
  with one certificate verification, with model-based projections up to
  $4{,}000\times$ faster than re-verifying the block's signatures at
  scale, under the Groth16 assumption.  Validators still verify $O(N)$
  signatures.

\item \textbf{Post-quantum migration by rotation.}  An account moves
  from Ed25519 to ML-DSA-44 keys by rotation and keeps its identifier;
  ML-DSA-44 transfers are post-quantum at the signature level, while
  Ed25519 keys and the Groth16 certificate are not (\cref{sec:quantum}).

\item \textbf{No per-transaction public keys.}  Public keys are
  installed once, so a transaction carries 113\,B of authorization data
  with Ed25519 keys and 2{,}469\,B with ML-DSA-44 keys.  This is about
  $1.5\times$ less than an ML-DSA-44 signature plus public key, and in
  Ed25519 mode it is larger than Solana's 96\,B
  (\cref{sec:onchain-data} and \cref{sec:bandwidth-tps}).
\end{enumerate}

Under our model, separating the identifier from the installed key lets
\sys place zero-knowledge proofs where they are paid once (birth) or
off the critical path (aggregation), keep native signatures on the hot
path, and change the signature scheme of an account without changing
its identity.

\section*{Acknowledgments}

The author thanks the ACE-GF research community for discussions on
identity--authorization separation and ZK proof system design.

\begingroup
\raggedright

\endgroup

\appendix
\section{Finality Timeline Comparison}
\label{app:finality-timeline}

\begin{figure}[h]
\centering
\begin{tikzpicture}[
  >=Stealth,
  timeline/.style={thick, ->},
  event/.style={fill=white, draw, rounded corners=2pt, font=\footnotesize,
    inner sep=3pt, align=center},
  marker/.style={circle, fill, inner sep=1.5pt},
]
\node[font=\small\bfseries, anchor=east] at (-0.3, 2.5) {ACE Runtime};
\draw[timeline] (0, 2.5) -- (13, 2.5);
\node[marker, fill=green!60!black] at (0.5, 2.5) {};
\node[event, above] at (0.5, 2.7) {tx submit};
\node[marker, fill=orange] at (2.0, 2.5) {};
\node[event, above] at (2.0, 2.7) {soft\\(400\,ms)};
\node[marker, fill=red!70!black] at (3.0, 2.5) {};
\node[event, above] at (3.4, 2.7) {\textbf{hard}\\(+ aggregation)};

\node[font=\small\bfseries, anchor=east] at (-0.3, 0.8) {Solana};
\draw[timeline] (0, 0.8) -- (13, 0.8);
\node[marker, fill=green!60!black] at (0.5, 0.8) {};
\node[event, below] at (0.5, 0.6) {tx submit};
\node[marker, fill=orange] at (2.0, 0.8) {};
\node[event, below] at (2.0, 0.6) {soft\\(400\,ms)};
\node[marker, fill=red!70!black] at (7.5, 0.8) {};
\node[event, below] at (7.5, 0.6) {\textbf{hard}\\(${\sim}$12\,s)};

\node[font=\small\bfseries, anchor=east] at (-0.3, -0.9) {Ethereum};
\draw[timeline] (0, -0.9) -- (13, -0.9);
\node[marker, fill=green!60!black] at (0.5, -0.9) {};
\node[event, below] at (0.5, -1.1) {tx submit};
\node[marker, fill=orange] at (2.5, -0.9) {};
\node[event, below] at (2.5, -1.1) {soft\\(12\,s)};
\node[marker, fill=red!70!black] at (12.5, -0.9) {};
\node[event, below] at (12.5, -1.1) {\textbf{hard}\\(${\sim}$15\,min)};

\draw[<->] (0, -2.3) -- (13, -2.3);
\node[font=\scriptsize, below] at (6.5, -2.4) {Time (not to scale)};

\end{tikzpicture}
\caption{Finality timeline comparison.  \sys reaches soft finality at
  ${\sim}400\;\mathrm{ms}$; its hard finality follows after the
  aggregation time, which is not measured.  Solana's hard finality takes
  ${\sim}12\;\mathrm{s}$ and Ethereum's ${\sim}15$\,min, under different
  definitions (\cref{tab:finality-compare}).}
\label{fig:finality-timeline}
\end{figure}

\section{ZK-ACE Circuit Constraint Summary}
\label{app:circuit}

\begin{table}[h]
\centering
\caption{ZK-ACE circuit sponge schedule (from \zkace~\cite{zk-ace}).}
\label{tab:circuit-constraints}
\begin{tabular}{@{}lcc@{}}
\toprule
\textbf{Constraint Group} & \textbf{Description} & \textbf{Permutations} \\
\midrule
Identity binding (C1) & $\idcom = \mathsf{S}_1(\mathsf{Enc}(\REV, \mathsf{s}_{\mathrm{id}}, d))$ & 3 \\
Target (C2) & $t = \mathsf{S}_t(\mathsf{Enc}(\REV, \mathsf{ctx}))$ & 2 + 1 squeeze \\
Replay value (C3) & $\mathit{rp} = \mathsf{S}_r(\mathsf{Enc}(\REV, \mathsf{ctx}, m, d, n))$ & 4 + 1 squeeze \\
Policy (optional) & $\mathit{target} \in \mathit{allowed\_set}$ & not in \zkace \\
\midrule
\textbf{Total} & & \textbf{11} \\
\bottomrule
\end{tabular}
\begin{flushleft}
\footnotesize Poseidon2 permutations; $253$ active AIR rows in total.
(C4) and (C5) are equality constraints and add no permutations.
\zkace measures this relation only as a Circle STARK
($347{,}094$\,B JSON envelope, ${\approx}339$\,KiB); in \sys it is the
birth proof.  The Groth16/BN254 figures used in this paper for the
aggregate certificate (256-byte proof, ${\sim}0.5\;\mathrm{ms}$
verification with 3 bilinear pairing operations) are assumptions, not
measurements (\cref{sec:proof-pipeline}).
\end{flushleft}
\end{table}

\section{ACE Block and Finality Certificate Format}
\label{app:block-format}

\begin{table}[h]
\centering
\caption{ACE block header fields.}
\label{tab:block-header}
\begin{tabular}{@{}llc@{}}
\toprule
\textbf{Field} & \textbf{Type} & \textbf{Size (bytes)} \\
\midrule
\texttt{slot\_number} & \texttt{u64} & 8 \\
\texttt{parent\_hash} & \texttt{[u8; 32]} & 32 \\
\texttt{state\_root} & \texttt{[u8; 32]} & 32 \\
\texttt{tx\_merkle\_root} & \texttt{[u8; 32]} & 32 \\
\texttt{attest\_merkle\_root} & \texttt{[u8; 32]} & 32 \\
\texttt{poh\_hash} & \texttt{[u8; 32]} & 32 \\
\texttt{leader\_id\_com} & \texttt{[u8; 32]} & 32 \\
\texttt{timestamp} & \texttt{u64} & 8 \\
\texttt{tx\_count} & \texttt{u32} & 4 \\
\midrule
\textbf{Total header} & & $\mathbf{{\sim}256}$ \\
\bottomrule
\end{tabular}
\end{table}

\begin{table}[h]
\centering
\caption{Finality certificate format (fixed size, block-size independent).}
\label{tab:fc-format}
\begin{tabular}{@{}llc@{}}
\toprule
\textbf{Field} & \textbf{Type} & \textbf{Size (bytes)} \\
\midrule
\texttt{block\_hash} & \texttt{[u8; 32]} & 32 \\
\texttt{slot\_number} & \texttt{u64} & 8 \\
\texttt{groth16\_proof} & \texttt{[u8; 256]} & 256 \\
\texttt{public\_inputs} & \texttt{Commitment([u8; 32])} & 32 \\
\midrule
\textbf{Total} & & $\mathbf{328}$ \\
\bottomrule
\end{tabular}
\end{table}

\end{document}